\RequirePackage{fix-cm}
\pdfoutput=1
\documentclass[a4paper,USenglish,runningheads]{llncs}

\usepackage[T1]{fontenc}
\usepackage[utf8]{inputenc}
\usepackage{babel}

\usepackage{graphicx}
\usepackage{wrapfig}
\usepackage{url}
\usepackage[usenames,dvipsnames]{xcolor}
\usepackage[protrusion=true]{microtype}
\usepackage{enumitem}
\usepackage[unicode=true,
    bookmarks=true,bookmarksnumbered=false,bookmarksopen=false,
    breaklinks=true,pdfborder={0 0 1},colorlinks=true]
    {hyperref}

\usepackage[mathscr,scaled=1.15]{urwchancal} 
\usepackage{mathtools}
\usepackage{amssymb}
\usepackage{stmaryrd}
\usepackage{centernot}
\usepackage[misc,geometry]{ifsym}
\usepackage[lined,boxed,commentsnumbered]{algorithm2e}
\usepackage{adjustbox}

\usepackage{tikz}
\usetikzlibrary{arrows,automata}
\usetikzlibrary{decorations}
\usetikzlibrary{decorations.pathmorphing}
\usetikzlibrary{arrows.meta}
\usetikzlibrary{shapes.geometric}
\usetikzlibrary{shapes.multipart}
\usetikzlibrary{shapes}
\usetikzlibrary{fadings}
\usetikzlibrary{fit}

\usepackage{pgf}
\usepackage{pgfplots}
\usepackage{pgfplotstable}
\pgfplotsset{compat=1.16}

\usepackage{wrapfig, tabularx,booktabs}

\usepackage{colortbl}

\usepackage{todonotes}
\usepackage{xstring}
\usepackage[strings]{underscore}

\usepackage[firstpage]{draftwatermark}  
\SetWatermarkAngle{0}
\SetWatermarkText{\raisebox{17cm}{%
  \hspace{0.7cm}%
  \href{https://doi.org/10.5281/zenodo.7870252}{\includegraphics{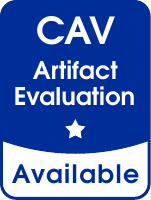}}%
  \hspace{9cm}%
  \includegraphics{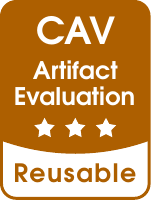}%
}}

\usepackage{ifthen}
\newboolean{fullversion}
\setboolean{fullversion}{true}

\allowdisplaybreaks
\spnewtheorem*{postponedProof}{Proof of}{\bfseries}{\itshape}

\hypersetup{pdftitle={Process Equivalence Problems as Energy Games},
    pdfauthor={Benjamin Bisping},
    pdfkeywords={Process equivalence spectrum, Minimization, Energy games, Bisimulation games}}
\makeatletter
\def\@citecolor{blue}%
\def\@urlcolor{blue}%
\def\@linkcolor{blue}%

\ifthenelse{\boolean{fullversion}}
  {\def\orcidID#1{{\href{http://orcid.org/#1}{\protect\raisebox{-1.25pt}{\protect\includegraphics{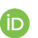}}}}}}{}
\makeatother

\SetAlCapSty{textnormal}
\SetArgSty{textnormal}
\SetProgSty{textnormal}
\SetFuncArgSty{textnormal}

\LinesNumbered
\DontPrintSemicolon
\SetKwProg{For}{for}{\string:}{}
\SetKwProg{Fn}{def}{\string:}{}
\SetKwProg{While}{while}{\string:}{}
\SetKwIF{If}{ElseIf}{Else}{if}{\string:}{elif}{else\string:}{}%

\global\long\def\defEq{\mathrel{\coloneqq}}%
\global\long\def\varname#1{\mathsf{#1}}%
\global\long\def\ccsIdentifier#1{\mathsf{#1}}%
\global\long\def\ccsInm#1{#1}%

\global\long\def\rel#1{\mathcal{#1}}%
\global\long\def\bigo{\mathcal{O}}%
\newcommand*{\relSize}[1]{\lvert\mathord{#1}\rvert}
\newcommand{\powerSet}[1]{\mathbf{2}^{#1}}

\newcommand{\nats}{\mathbb{N}}
\newcommand{\ints}{\mathbb{Z}}
\newcommand{\domain}{\operatorname{\mathrm{dom}}}
\newcommand*{\vectorComponents}[2][n]{({#2}_1,\ldots,{#2}_{#1})}
\newcommand*{\Min}{\operatorname{\mathrm{Min}}}

\newcommand*{\gameMoveX}[1]{\mathrel{\smash{\xrightarrowtail{\scriptscriptstyle#1}}}}%
\newcommand*{\gameMove}{\gameMoveX{\hspace*{0.5em}}}%
\newcommand*{\game}{\mathcal{G}}%
\newcommand*{\gameSpectroscopy}{\mathcal{G}_\triangle}%
\newcommand*{\gameSpectroscopyClever}{\mathcal{G}_\blacktriangle}%
\newcommand*{\attackerPos}[2][]{{{\color{gray}(}#2{\color{gray})}}_\mathtt{a}^{\color{gray}\smash{\scriptscriptstyle#1}}}
\newcommand*{\defenderPos}[2][]{{{\color{gray}(}#2{\color{gray})}}_\mathtt{d}^{\color{gray}\smash{\scriptscriptstyle#1}}}

\newcommand*{\attackerSubscript}{{\operatorname{a}}}
\newcommand*{\defenderSubscript}{{\operatorname{d}}}
\newcommand*{\energies}{\mathbf{En}}
\newcommand*{\energyUpdates}{\mathbf{Up}}%
\newcommand*{\energyUpdate}{\mathsf{upd}}%
\newcommand*{\energyUpdateInv}{\mathsf{upd}^{-1}}%
\newcommand*{\updMin}[1]{\mathtt{min}_{\{\!#1\!\}}}
\newcommand*{\energyLevel}{\mathsf{EL}}%
\newcommand*{\attackerWin}{\mathsf{Win}_\attackerSubscript}
\newcommand*{\attackerWinMin}{\attackerWin^{\scriptscriptstyle\min}}
\newcommand*{\defenderWinMax}{\mathsf{Win}_\defenderSubscript^{\scriptscriptstyle\max}}

\newcommand*{\proc}{\mathcal{P}}
\newcommand*{\system}{\mathcal{S}}%
\newcommand*{\action}[1]{\mathit{#1}}
\newcommand*{\actions}{\Sigma}%
\newcommand*{\step}[1]{\mathrel{\smash{\xrightarrow{#1}}}}%
\newcommand*{\initials}{\mathcal{I}}

\newcommand{\hml}{\mathsf{HML}}
\newcommand{\hmlA}{\hml[\actions]}
\newcommand{\hmlObs}[1]{\langle#1\rangle}

\newcommand{\hmlAnd}[2]{{\bigwedge_{#1 \in #2}}}
\newcommand{\hmlAndS}{{\bigwedge}}
\newcommand{\hmlTrue}{\mathsf{T}}
\newcommand{\hmlNeg}{\neg}
\newcommand{\hmlSemantics}[3]{{\llbracket #1 \rrbracket}^{#2}_{#3}}
\newcommand{\hmlStrategies}{\mathsf{Strat}}

\newcommand{\expr}{\mathsf{expr}}

\let\obs\undefined
\newcommand*{\eqName}[1]{\mathrm{#1}}
\newcommand*{\obs}[1]{\mathcal{O}_{\mathrm{#1}}}
\newcommand*{\bEquiv}[1]{\sim_{\mathrm{#1}}}
\newcommand*{\bPreord}[1]{\preceq_{\mathrm{#1}}}

\newcommand{\refDef}[1]{Definition~\ref{#1}}
\newcommand{\refExample}[1]{Example~\ref{#1}}

\newcommand{\refThm}[1]{Theorem~\ref{#1}}
\newcommand{\refProp}[1]{Proposition~\ref{#1}}
\newcommand{\refLem}[1]{Lemma~\ref{#1}}

\newcommand{\refSec}[1]{Section~\ref{#1}}
\newcommand{\refSubsec}[1]{Subsection~\ref{#1}}
\newcommand{\refFig}[1]{Figure~\ref{#1}}

\newcommand{\refAlgo}[1]{Algorithm~\ref{#1}}
\newcommand{\refLine}[1]{Line~\ref{#1}}

\newcommand*{\ie}{\text{i.e.\,}}
\newcommand*{\cf}{\text{cf.\,}}

\colorlet{darkgreen}{green!60!black}
\colorlet{darkred}{red!50!black}

\makeatletter
\newbox\xrat@below
\newbox\xrat@above
\newcommand{\xrightarrowtail}[2][]{%
  \setbox\xrat@below=\hbox{\ensuremath{\scriptstyle #1}}%
  \setbox\xrat@above=\hbox{\ensuremath{\scriptstyle #2}}%
  \pgfmathsetlengthmacro{\xrat@len}{max(\wd\xrat@below,\wd\xrat@above)+.6em}%
  \mathrel{\tikz [>->,baseline=-.58ex,line width=0.43pt]
                 \draw (0,0) -- node[below=-2pt] {\box\xrat@below}
                                node[above=-2pt] {\box\xrat@above}
                       (\xrat@len,0) ;}}
\makeatother

\makeatletter
\newcommand*{\centerfloat}{%
  \parindent \z@
  \leftskip \z@ \@plus 1fil \@minus \textwidth
  \rightskip\leftskip
  \parfillskip \z@skip}
\makeatother

\begin{document}
\title{Process Equivalence Problems as Energy Games}
%
%
\author{Benjamin Bisping \orcidID{0000-0002-0637-0171}}
\institute{Technische Universit\"at Berlin, Berlin, Germany\\
\email{benjamin.bisping@tu-berlin.de}~\textbullet~\href{https://bbisping.de}{bbisping.de}}
%
%
\maketitle              
\begin{abstract}

  We characterize all common notions of behavioral equivalence by \emph{one} 6-dimensional energy game,
  where energies bound capabilities of an attacker trying to tell processes apart.
  The def\-en\-der-win\-ning initial credits exhaustively determine which preorders and equivalences from the (strong) linear-time--branching-time spectrum relate processes.

  \qquad
  The time complexity is exponential, which is optimal due to trace equivalence being covered.
  This complexity improves drastically on our previous approach for deciding groups of equivalences where exponential sets of distinguishing HML formulas are constructed on top of a su\-per-ex\-po\-nen\-tial reachability game.
  In experiments using the VLTS benchmarks, the algorithm performs on par with the best similarity algorithm.


\end{abstract}
\section{Introduction}

\begin{figure}[b!]
  \begin{adjustbox}{center, trim={10mm 0 20mm 5mm}}
    \begin{tikzpicture}[scale=.9, transform shape,->,auto,node distance=1.7cm, rel/.style={dashed,font=\it, blue},
        ext/.style={line width=1pt},
        internal/.style={draw=gray}]
      \node (P0) {$\ccsIdentifier{Pe}$};
      \node (P1)[below left of=P0] {$\cdot$};
      \node (P2)[below right of=P0] {$\cdot$};
      \node (P3)[below left of=P1] {$\cdot$};
      \node (P4)[below right of=P1] {$\cdot$};
      \node (P5)[below right of=P2] {$\cdot$};
      \node (P6)[below left of=P4] {$\cdot$};
      \node (P7)[below right of=P4] {$\cdot$};
      \node (P8)[below left of=P3] {$\color{BrickRed}\circ$};
      \node (P9)[below of=P3, node distance=1.2cm] {$\cdot$};
      \node (P10)[below of=P5, node distance=1.2cm] {$\cdot$};
      \node (P11)[below right of=P5] {$\color{BlueViolet}\circ$};
      \node (P12)[below left of=P6] {$\color{BrickRed}\circ$};
      \node (P13)[below right of=P7] {$\color{BlueViolet}\circ$};
      \node (P14)[below left of=P9] {$\color{BrickRed}\circ$};
      \node (P15)[below right of=P10] {$\color{BlueViolet}\circ$};

      \path
      (P0) edge [internal] node {} (P1)
      (P0) edge [internal] node {} (P2)
      (P1) edge [internal] node {} (P3)
      (P1) edge [internal] node {} (P4)
      (P2) edge [internal] node {} (P4)
      (P2) edge [internal] node {} (P5)
      (P3) edge node[swap, pos=.3] {$\action{ec_A}$} (P8)
      (P3) edge [internal] node {} (P9)
      (P4) edge [internal] node {} (P6)
      (P4) edge [internal] node {} (P7)
      (P5) edge [internal] node {} (P10)
      (P5) edge [pos=.3] node {$\action{ec_B}$} (P11)
      (P6) edge node[swap, pos=.3] {$\action{ec_A}$} (P12)
      (P7) edge [pos=.3] node {$\action{ec_B}$} (P13)
      (P8) edge [internal] node {} (P14)
      (P8) edge [bend left=70] node {$\action{lc_A}$} (P0)
      (P9) edge [swap,pos=.3] node {$\action{ec_A}$} (P14)
      (P9) edge [internal] node {} (P6)
      (P10) edge [pos=.3] node {$\action{ec_B}$} (P15)
      (P10) edge [internal] node {} (P7)
      (P11) edge [internal] node {} (P15)
      (P11) edge [bend right=70, swap] node {$\action{lc_B}$} (P0)
      (P12) edge [bend left=13, swap, pos=.3] node {$\action{lc_A}$} (P5)
      (P13) edge [bend right=13, pos=.3] node {$\action{lc_B}$} (P3)
      (P14) edge [internal] node {} (P12)
      (P14) edge [bend left=85] node {$\action{lc_A}$} (P2)
      (P15) edge [internal] node {} (P13)
      (P15) edge [bend right=100, swap] node {$\action{lc_B}$} (P1)
      ;

      \node (M0)[left of=P0, node distance=6.5cm] {$\ccsIdentifier{Mx}$};
      \node (M1)[below left of=M0, node distance=2.5cm] {$\color{BrickRed}\circ$};
      \node (M2)[below right of=M0, node distance=2.5cm] {$\color{BlueViolet}\circ$};

      \path
      (M0) edge [bend left=15] node {$\action{ec_A}$} (M1)
      (M0) edge [bend right=15, swap] node {$\action{ec_B}$} (M2)
      (M1) edge [bend left=15] node {$\action{lc_A}$} (M0)
      (M2) edge [bend right=15, swap] node {$\action{lc_B}$} (M0)
      ;
    \end{tikzpicture}
  \end{adjustbox}
  \caption{A specification of mutual exclusion $\ccsIdentifier{Mx}$, and Peterson's protocol $\ccsIdentifier{Pe}$.}
  \label{fig:ts-peterson}
\end{figure}

Many verification tasks can be understood along the lines of ``how equivalent'' two models are.
\refFig{fig:ts-peterson} replicates a standard example, known for instance from the textbook \emph{Reactive Systems}~\cite{ails2007reactiveSystems}:
A specification of mutual exclusion $\ccsIdentifier{Mx}$ as two alternating users $A$ and $B$ entering their critical section $\action{ec_A} / \action{ec_B}$ and leaving $\action{lc_A} / \action{lc_B}$ before the other may enter;
and the transition system of Peterson's~\cite{peterson1981mythsMutex} mutual exclusion algorithm $\ccsIdentifier{Pe}$, minimized by weak bisimilarity, with internal steps $\color{gray}\step{}$ due to the coordination that needs to happen.
For $\ccsIdentifier{Pe}$ to faithfully implement mutual exclusion, it should behave somewhat similarly to $\ccsIdentifier{Mx}$.

Semantics in concurrent models must take nondeterminism into account.
Setting the degree to which nondeterminism counts induces equivalence notions with subtle differences:
$\ccsIdentifier{Pe}$ and $\ccsIdentifier{Mx}$ \emph{weakly simulate} each other, meaning that a tree of options from one process can be matched by a similar tree of the other.
This implies that they have the same \emph{weak traces}, that is, matching paths.
However, they are not weakly \emph{bi-}similar, which would require a higher degree of symmetry than mutual simulation, namely, matching absence of options.
There are many more such notions.
Van Glabbeek's \emph{linear-time--branching-time spectrum}~\cite{glabbeek1990ltbt1} (\cf~\refFig{fig:ltbt-spectrum}) brings order to the hierarchy of equivalences.
But it is notoriously difficult to navigate.
In our example, one might wonder: Are there notions relating the two \emph{besides} mutual simulation?

Our recent algorithm for \emph{linear-time--branching-time spectroscopy} by Bisping, Nestmann, and Jansen~\cite{bisping2021ltbtsTacas,bjn2022decidingAllBehavioralEqs} is capable of answering equivalence questions for finite-state systems by \emph{deciding the spectrum of behavioral equivalences in one go}.
In theory, that is.
In practice, the algorithm of~\cite{bjn2022decidingAllBehavioralEqs} runs out of memory when applied to the weak transition relation of even small examples like $\ccsIdentifier{Pe}$.
The reason for this is that saturating transition systems with the closure of weak steps adds a lot of nondeterminism.
For instance, $\ccsIdentifier{Pe}$ may reach 10 different states by internal steps ($\color{gray}\step{}^*$).
The spectroscopy algorithm of~\cite{bjn2022decidingAllBehavioralEqs} builds a bisimulation game where the defender wins if the game starts at a pair of equivalent processes.
To allow all attacks relevant for the spectrum, the \cite{bjn2022decidingAllBehavioralEqs}-game must consider partitionings of state sets reached through nondeterminism.
There are 115,975 ways of partitioning 10~objects.
As a consequence, the game graph of~\cite{bjn2022decidingAllBehavioralEqs} comparing $\ccsIdentifier{Pe}$ and $\ccsIdentifier{Mx}$ has 266,973~game positions.
On top of each postion,~\cite{bjn2022decidingAllBehavioralEqs} builds sets of distinguishing formulas of Hennessy--Milner modal logic (HML)~\cite{hm1980hml,glabbeek1990ltbt1} with minimal expressiveness.
These sets may grow exponentially.
Game over!

\subsubsection{Contributions.}

In this paper, we adapt the spectroscopy approach of~\cite{bisping2021ltbtsTacas,bjn2022decidingAllBehavioralEqs} to render small verification instances like $\ccsIdentifier{Pe}/\ccsIdentifier{Mx}$ feasible.
The key ingredients that will make the difference are:
understanding the spectrum purely through \emph{depth-properties of HML formulas};
using \emph{multidimensional energy games}~\cite{fjls2011energyGamesMulti} instead of reachability games;
and exploiting the considered spectrum to drastically \emph{reduce the branching-degree} of the game as well as the height of the energy lattice.
\refFig{fig:overview} lays out the algorithm with pointers to key parts of this paper.

\begin{figure}[t]
  \centering
  \begin{tikzpicture}[->,auto,node distance=3.6cm,
    algstep/.style={minimum width=2.5cm, minimum height=1cm,rectangle,draw=black,align=center,rounded corners}]
    \node[algstep] (TS) {Transition\\system};
    \node[algstep, below of=TS, node distance=1.7cm] (SEG) {Spectroscopy\\energy game};
    \node[algstep, right of=SEG] (WB) {Attacker win-\\ning budgets};
    \node[algstep, above of=WB, node distance=1.7cm] (PE) {Set of preorders\\\& equivalences};
    \node[algstep, gray, right of=PE] (DF) {Distinguishing\\HML formulas};
    \path
      (TS) edge node {\ref{subsec:spectroscopy-game} / \ref{subsec:more-clever}} (SEG)
      (SEG) edge node {\ref{subsec:computation-winning}} (WB)
      (WB) edge[-, decorate, decoration=snake, gray] node[swap] {\ref{subsec:game-correctness}} (DF)
      (PE) edge[-, decorate, decoration=snake, gray] node {\ref{subsec:notions-equivalence}} (DF)
      (WB) edge node {\ref{subsec:obtaining-eqs}} (PE);
  \end{tikzpicture}
  \caption{Overview of the computations $\rightarrow$ and correspondences $\sim$ we will discuss.}
  \label{fig:overview}
\end{figure}

\begin{itemize}
  \item \refSubsec{subsec:notions-equivalence} explains how the \emph{linear-time--branching-time spectrum can be understood in terms of six dimensions of HML expressiveness}, and \refSubsec{subsec:energy-games} introduces a class of \emph{declining energy games} fit for our task.
  \item In \refSubsec{subsec:spectroscopy-game}, we describe the novel \emph{spectroscopy energy game}, and, in \refSubsec{subsec:game-correctness}, \emph{prove it to characterize all notions of equivalence definable by the six dimensions}.
  \item \refSubsec{subsec:more-clever} shows that a \emph{more clever game with only linear branching-factor} still covers the spectrum.
  \item \refSubsec{subsec:computation-winning} provides an algorithm to \emph{compute winning initial energy levels for declining energy games with $\updMin{\dots}$}, which enables \emph{decision of the whole considered spectrum in $\smash{2^{\bigo(\relSize{\proc})}}$} for systems with $\relSize{\proc}$ processes (\refSubsec{subsec:complexity}).
  \item In \refSubsec{subsec:obtaining-eqs}, we add fine print on \emph{how to obtain equivalences and distinguishing formulas} in the algorithm.
  \item \refSec{sec:implementation} compares to \cite{bjn2022decidingAllBehavioralEqs} and \cite{rt2010efficientSimulation} through \emph{experiments with the widely used VLTS benchmark suite}~\cite{garavel2017vlts}.
  The experiments also reveal insights about the suite itself.
\end{itemize}

\section{Distinctions and Equivalences in Transition Systems}
\label{sec:background}

Two classic concepts of system analysis form the background of this paper:
\emph{Hen\-ne\-ssy--Milner logic} (HML) interpreted over \emph{transition systems} goes back to Hennessy and Milner~\cite{hm1980hml} investigating observational equivalence in operational semantics. Van Glabbeek's \emph{linear-time--branching-time spectrum}~\cite{glabbeek1990ltbt1} arranges all common notions of equivalence as a hierarchy of HML sublanguages.

\subsection{Transition Systems and Hennessy--Milner Logic}
\begin{definition}[Labeled transition system]
  \label{def:transition-system}
  A \emph{labeled transition system} is a tuple $\system=(\proc,\actions,\step{})$ where $\proc$ is the set of \emph{processes}, $\actions$ is the set of \emph{actions}, and ${\step{}}\subseteq \proc\times\actions\times \proc$ is the \emph{transition relation}.

  By $\initials(p)$ we denote the \emph{actions enabled initially} for a process $p \in \proc$, that is, $\initials(p) \mathrel{\defEq} \{a \in \actions \mid \exists p' \ldotp p \step{a} p'\}$.
  We lift the steps to sets with $P \step{a} P'$ iff $P' = \{p' \mid \exists p \in P \ldotp p \step{a} p' \}$.
\end{definition}

\noindent
Hennessy--Milner logic expresses \emph{observations} that one may make on such a system.
The set of formulas true of a process offers a denotation for its semantics.

\definecolor{obsColor}{RGB}{218, 0, 99}
\definecolor{conjColor}{RGB}{242, 75, 42}
\definecolor{mainPosColor}{RGB}{200, 180, 0}
\definecolor{otherPosColor}{RGB}{143, 209, 79}
\definecolor{negObsColor}{RGB}{18, 205, 212}
\definecolor{negsColor}{RGB}{101, 44, 179}

\begin{figure}[t]
  \centering
  \begin{tikzpicture}[auto,node distance=2.02cm,align=center]

    \node (B){bisimulation $\eqName{B}$\\$(\textcolor{obsColor}{\infty}, \textcolor{conjColor}{\infty}, \textcolor{mainPosColor}{\infty}, \textcolor{otherPosColor}{\infty}, \textcolor{negObsColor}{\infty}, \textcolor{negsColor}{\infty})$};
    \node (S2)[below of=B, node distance=1.2cm]{$2$-nested simulation $\eqName{2S}$\\$(\textcolor{obsColor}{\infty}, \textcolor{conjColor}{\infty}, \textcolor{mainPosColor}{\infty}, \textcolor{otherPosColor}{\infty}, \textcolor{negObsColor}{\infty}, \textcolor{negsColor}{1})$};
    \node (RS)[below left of=S2]{ready simulation $\eqName{RS}$~\\$(\textcolor{obsColor}{\infty}, \textcolor{conjColor}{\infty}, \textcolor{mainPosColor}{\infty}, \textcolor{otherPosColor}{\infty}, \textcolor{negObsColor}{1}, \textcolor{negsColor}{1})$};
    \node (RT)[below right of=RS]{readiness traces $\eqName{RT}$\\$(\textcolor{obsColor}{\infty}, \textcolor{conjColor}{\infty}, \textcolor{mainPosColor}{\infty}, \textcolor{otherPosColor}{1}, \textcolor{negObsColor}{1}, \textcolor{negsColor}{1})$};
    \node (FT)[below left of=RT]{failure traces $\eqName{FT}$\\$(\textcolor{obsColor}{\infty}, \textcolor{conjColor}{\infty}, \textcolor{mainPosColor}{\infty}, 0, \textcolor{negObsColor}{1}, \textcolor{negsColor}{1})$};
    \node (R)[below right of=RT]{readiness $\eqName{R}$\\$(\textcolor{obsColor}{\infty}, \textcolor{conjColor}{2}, \textcolor{mainPosColor}{1}, \textcolor{otherPosColor}{1}, \textcolor{negObsColor}{1}, \textcolor{negsColor}{1})$};
    \node (PF)[above right of=R]{~possible futures $\eqName{PF}$\\$(\textcolor{obsColor}{\infty}, \textcolor{conjColor}{2}, \textcolor{mainPosColor}{\infty}, \textcolor{otherPosColor}{\infty}, \textcolor{negObsColor}{\infty}, \textcolor{negsColor}{1})$};
    \node (IF)[below right of=R]{impossible futures $\eqName{IF}$\\$(\textcolor{obsColor}{\infty}, \textcolor{conjColor}{2}, 0, 0, \textcolor{negObsColor}{\infty}, \textcolor{negsColor}{1})$};
    \node (S)[left of=FT, node distance=3cm]{simulation $\eqName{1S}$\\$(\textcolor{obsColor}{\infty}, \textcolor{conjColor}{\infty}, \textcolor{mainPosColor}{\infty}, \textcolor{otherPosColor}{\infty}, 0, 0)$};
    \node (RV)[below left of=R]{revivals $\eqName{RV}$\\$(\textcolor{obsColor}{\infty}, \textcolor{conjColor}{2}, \textcolor{mainPosColor}{1}, 0, \textcolor{negObsColor}{1}, \textcolor{negsColor}{1})$};
    \node (F)[below right of=RV]{failures $\eqName{F}$\\$(\textcolor{obsColor}{\infty}, \textcolor{conjColor}{2}, 0, 0, \textcolor{negObsColor}{1}, \textcolor{negsColor}{1})$};
    \node (T)[below left of=F]{traces $\eqName{T}$\\$(\textcolor{obsColor}{\infty}, \textcolor{conjColor}{1},0,0,0,0)$};
    \node (E)[below of=T, node distance=1.2cm]{enabledness $\eqName{E}$\\$(\textcolor{obsColor}{1}, \textcolor{conjColor}{1},0,0,0,0)$};

    \path
    (B) edge (S2)
    (S2) edge (RS)
    (S2) edge (PF)
    (RS) edge (S)
    (RS) edge (RT)
    (RT) edge (R)
    (RT) edge (FT)
    (PF) edge (R)
    (PF) edge (IF)
    (S) edge (T)
    (FT) edge (RV)
    (R) edge (RV)
    (IF) edge (F)
    (RV) edge (F)
    (F) edge (T)
    (T) edge (E)
    ;
  \end{tikzpicture}
  \caption{Hierarchy of equivalences/preorders becoming finer towards the top.\label{fig:ltbt-spectrum}
  }
\end{figure}

\begin{definition}[Hennessy--Milner logic]
  \label{def:hml}
  The \emph{syntax of Hennessy--Milner logic} over a set $\actions$ of actions, $\hmlA$, is defined by the grammar:
  \begin{align*}
    \varphi {} ::= {} & \hmlObs{a}\varphi & \text{with } a \in \actions \\
        | \quad & \hmlAndS \{\psi, \psi, ...\} \\
    \psi {} ::= {}    & \hmlNeg \varphi \mid \varphi.
  \end{align*}
  Its semantics $\;\smash{\hmlSemantics{\;\cdot\;}{\system}{}}$ over a transition system $\system=(\proc,\actions,\step{})$ is given as the set of processes where a formula ``is true'' by:
  \begin{align*}
    \hmlSemantics{\hmlObs{a}\varphi}{\system}{} \defEq &\;
      \{p \in \proc \mid
        \exists p' \in \hmlSemantics{\varphi}{\system}{} \ldotp p \step{a} p'\}
    \\[0.25\baselineskip]
    \smash{\hmlSemantics{\hmlAnd{i}{I}\psi_i}{\system}{}} \defEq &\;
      \bigcap \{ \hmlSemantics{\psi_i}{\system}{} \mid
        i \in I \land
        \nexists \varphi \ldotp \psi_i = \hmlNeg \varphi \} \\
      \setminus &\;
      \bigcup \{ \hmlSemantics{\varphi}{\system}{} \mid
        \exists i \in I \ldotp
        \psi_i = \hmlNeg \varphi \}.
  \end{align*}
\end{definition}
HML basically extends propositional logic with a modal observation operation.
Conjunctions then bound trees of future behavior.
Positive conjuncts mean lower bounds, negative ones impose upper bounds.
For the scope of this paper, finite bounds suffice, \ie, conjunctions are finite-width.
The empty conjunction $\hmlTrue \defEq \hmlAndS \varnothing$ is usually omitted in writing.

We use Hennessy--Milner logic to capture \emph{differences} between program behaviors.
Depending on how much of its expressiveness we use, different notions of equivalence are characterized.

\begin{definition}[Distinguishing formulas and preordering languages]
  \label{def:distinguishing-formula}
  A formula $\varphi \in \hmlA$ is said to \emph{distinguish} two processes $p,q \in \proc$ iff $p \in \hmlSemantics{\varphi}{\system}{}$ and $q \notin \smash{\hmlSemantics{\varphi}{\system}{}}$\!.
  A sublanguage of Hennessy--Milner logic, $\obs{\mathit X} \subseteq \hmlA$, either distinguishes two processes, $p \not\bPreord{\mathit X} q$, if it contains a distinguishing formula, or preorders them otherwise.
  If processes are preordered in both directions, $p \bPreord{\mathit X} q$ and $q \bPreord{\mathit X} p$, then they are considered $X\!$-equivalent, $p \bEquiv{\mathit X} q$.
\end{definition}

\noindent
\refFig{fig:ltbt-spectrum} charts the \emph{linear-time--branching-time spectrum}.
If processes are preordered/equated by one notion of equivalence, they also are preordered/equated by every notion below.
We will later formally characterize the notions through \refProp{prop:language-prices}.
For a thorough presentation, we point to~\cite{glabbeek2001ltbtsiReport}.
For those familiar with the spectrum, the following example serves to refresh memories.

\begin{figure}[t]
  \center
  \begin{tikzpicture}[transform shape,->,auto,node distance=2cm, rel/.style={dashed,font=\it, blue},
      ext/.style={line width=1pt},
      internal/.style={draw=gray}]
    \node (S0)[minimum size=.7cm] {$\ccsIdentifier{S}$};
    \node (Div)[minimum size=.7cm, below right of=S0] {$\ccsIdentifier{Div}$};
    \node (S1)[minimum size=.7cm, above right of=Div] {$\ccsIdentifier{S'}$};

    \path
    (S0) edge [internal, loop left] node {$\color{gray}\tau$} (S0)
    (S0) edge [bend right=10, internal, swap] node {$\color{gray}\tau$} (Div)
    (S0) edge [bend left=10] node {$\action{ec_A}$} (Div)
    (S1) edge [internal, loop right] node {$\color{gray}\tau$} (S1)
    (S1) edge [bend left=10] node {$\action{ec_A}$} (Div)
    (Div) edge [internal, loop right] node {$\color{gray}\tau$} (Div)
    ;
  \end{tikzpicture}
  \caption{Example system of internal decision $\color{gray}\step{\tau}$ against an action $\step{\action{ec_A}}$.}
  \label{fig:ts-example}
\end{figure}

\begin{example}
  \label{exa:distinguishing-formulas}
  \refFig{fig:ts-example} shows a tiny slice of the weak-step-saturated version of our initial example from \refFig{fig:ts-peterson} that is at the heart of why $\ccsIdentifier{Pe}$ and $\ccsIdentifier{Mx}$ are not bisimula\-tion-equivalent.
  The difference between $\ccsIdentifier{S}$ and $\ccsIdentifier{S'}$ is that $\ccsIdentifier{S}$ can internally transition to $\ccsIdentifier{Div}$ (labeled $\color{gray}\step{\tau}$) without ever performing an $\action{ec_A}$ action.
  We can express this difference by the formula
  $\varphi_\ccsIdentifier{S} \defEq \hmlObs{\tau}\hmlAndS\{\hmlNeg\hmlObs{ec_A}\}$,
  meaning ``after $\tau$, $\action{ec_A}$ may be impossible.''
  It is true for $\ccsIdentifier{S}$, but not for $\ccsIdentifier{S'}$.
  Knowing a distinguishing formula means that $\ccsIdentifier{S}$ and $\ccsIdentifier{S'}$ cannot be bisimilar by the Hennessy--Milner theorem.
  The formula $\varphi_\ccsIdentifier{S}$ is called a \emph{failure} (or \emph{refusal}) as it specifies a set of actions that are disabled after a trace.
  In the other direction of comparison, the negation
  $\varphi_\ccsIdentifier{S'} \defEq \hmlAndS\{ \hmlNeg \hmlObs{\tau}\hmlAndS\{ \hmlNeg \hmlObs{ec_A} \} \}$
  distinguishes $\ccsIdentifier{S'}$ from $\ccsIdentifier{S}$.
  The differences between the two processes cannot be expressed in HML without negation.
  Therefore the processes are simulation-equivalent, or \emph{similar}, as similarity is characterized by the positive fragment of HML.
\end{example}

\subsection{Price Spectra of Behavioral Equivalences}
\label{subsec:notions-equivalence}

For algorithms exploring the linear-time--branching-time spectrum, it is convenient to have a representation of the spectrum in terms of numbers or ``prices'' of formulas as in~\cite{bjn2022decidingAllBehavioralEqs}.
We, here, use six dimensions to characterize the notions of equivalence depicted in \refFig{fig:ltbt-spectrum}.
The numbers define the HML observation languages that characterize the very preorders/equivalences. Intuitively, the colorful numbers mean: (\textcolor{obsColor}{1})~Formula modal depth of observations. (\textcolor{conjColor}{2})~Formula nesting depth of conjunctions. (\textcolor{mainPosColor}{3})~Maximal modal depth of deepest positive clauses in conjunctions. (\textcolor{otherPosColor}{4})~Maximal modal depth of the other positive clauses in conjunctions. (\textcolor{negObsColor}{5})~Maximal modal depth of negative clauses in conjunctions. (\textcolor{negsColor}{6})~Formula nesting depth of negations. More formally:

\begin{definition}[Energies]
  \label{def:energies}
  We denote as \emph{energies}, $\energies$, the set of $N$-dimensional vectors $(\nats)^N$\!, and as \emph{extended energies}, $\energies_\infty$, the set $(\nats\cup\{\infty\})^N$.

  We compare energies component-wise, \ie, $\vectorComponents[N]{e} \leq \vectorComponents[N]{f}$ iff $e_i \leq f_i$ for each $i$.
  Least upper bounds $\sup$ are defined as usual as component-wise supremum,
  as are greatest lower bounds $\inf$.
\end{definition}


\begin{figure}[t]
  \centering
  \begin{tikzpicture}
    \node(Obs1) {$\hmlObs\tau$};
    \node(Conj0)[right=2cm of Obs1.center] {$\land$};
    \node(ObsEC1)[above right=1cm and 3cm of Conj0.center] {$\hmlObs{ec_A}$};
    \node(ObsLC1)[right=2cm of ObsEC1.center] {$\hmlObs{lc_A}$};
    \node(And1)[right=2cm of ObsLC1.center] {$\land$};
    \node(ObsEC2)[right=3cm of Conj0.center] {$\hmlObs{\tau}$};
    \node(And2)[right=2cm of ObsEC2.center] {$\land$};
    \node(Neg)[below right=1cm and 3cm of Conj0.center] {$\neg$};
    \node(NegObsEC)[right=2cm of Neg.center] {$\hmlObs{ec_B}$};
    \node(And3)[right=2cm of NegObsEC.center] {$\land$};

    \node(E1) [above left=10mm and 1mm of Obs1.center]{$e_{\textcolor{obsColor}{1}} = 3$};
    \node(E2) [above left=6mm and 1mm of Obs1.center]{$e_{\textcolor{conjColor}{2}} = 2$};
    \node(E3) [above right=3mm and 0mm of Conj0.center]{$e_{\textcolor{mainPosColor}{3}} = 2$};
    \node(E4) [above right=3mm and 20mm of Conj0.center]{$e_{\textcolor{otherPosColor}{4}} = 1$};
    \node(E5) [below right=2mm and 20mm of Conj0.center]{$e_{\textcolor{negObsColor}{5}} = 1$};
    \node(E6) [below left=2mm and 1mm of Obs1.center]{$e_{\textcolor{negsColor}{6}} = 1$};

    \draw[draw=obsColor] ([yshift=12mm]Obs1.center) circle[fill, radius=1pt] to ([yshift=12mm]Conj0.center) to[out=0,in=180] ([yshift=12mm]ObsEC1.center) circle[fill, radius=1pt] to ([yshift=12mm]ObsLC1.center) circle[fill, radius=1pt] to ([yshift=12mm]And1.center);
    \draw[draw=conjColor] ([yshift=8mm]Obs1.center) to ([yshift=8mm]Conj0.center) circle[fill, radius=1pt] to[out=0,in=180] ([yshift=8mm]ObsEC1.center) to ([yshift=8mm]ObsLC1.center) to ([yshift=8mm]And1.center) circle[fill, radius=1pt];
    \draw[draw=mainPosColor] ([yshift=2mm]Conj0) to[out=0,in=180] ([yshift=4mm]ObsEC1.center) circle[fill, radius=1pt] to ([yshift=4mm]ObsLC1.center) circle[fill, radius=1pt] to ([yshift=4mm]And1.center);
    \draw[draw=otherPosColor] ([yshift=2mm]Conj0) to[out=0,in=180] ([yshift=4mm]ObsEC2.center) circle[fill, radius=1pt] to ([yshift=4mm]And2.center);
    \draw[draw=negObsColor] ([yshift=-2mm]Conj0) to[out=0,in=180] ([yshift=4mm]Neg.center) to  ([yshift=4mm]NegObsEC.center) circle[fill, radius=1pt] to ([yshift=4mm]And3.center);
    \draw[draw=negsColor] ([yshift=-4mm]Obs1.center) to ([yshift=-4mm]Conj0.center) to[out=0,in=180] ([yshift=-4mm]Neg.center) circle[fill, radius=1pt] to ([yshift=-4mm]NegObsEC.center) to ([yshift=-4mm]And3.center);

    \draw[thick] (Obs1) to (Conj0) to[out=0,in=180] (ObsEC1) to (ObsLC1) to (And1);
    \draw[thick] (Conj0) to (ObsEC2) to (And2);
    \draw[thick] (Conj0) to[out=0,in=180] (Neg) to (NegObsEC) to (And3);
  \end{tikzpicture}
  \caption{Pricing $e$ of formula $\hmlObs{\tau}\hmlAndS\{\hmlObs{ec_A}\hmlObs{lc_A}\hmlTrue, \hmlObs{\tau}\hmlTrue, \hmlNeg\hmlObs{ec_B}\hmlTrue\}$.\label{fig:formula-price}}
\end{figure}

\begin{definition}[Formula prices]
  \label{def:formula-prices}
  The \emph{expressiveness price} $\expr \colon \hmlA \rightarrow (\nats)^6$ of a formula interpreted as $6$-dimensional energies is defined recursively by:
  \begin{align*}
    \expr(\hmlObs{a}\varphi) & \defEq
    \begin{pmatrix}
      \textcolor{obsColor}{1\,+\,} \expr_1(\varphi)\\
      \expr_2(\varphi)\\
      \expr_3(\varphi)\\
      \expr_4(\varphi)\\
      \expr_5(\varphi)\\
      \expr_6(\varphi)
    \end{pmatrix}
    \qquad
    \expr(\hmlNeg\varphi) \defEq
    \begin{pmatrix}
      \expr_1(\varphi)\\
      \expr_2(\varphi)\\
      \expr_3(\varphi)\\
      \expr_4(\varphi)\\
      \expr_5(\varphi)\\
      \textcolor{negsColor}{1\,+\,} \expr_6(\varphi)
    \end{pmatrix}
    \\
    \textstyle\expr(\bigwedge\limits_{i \in I}\psi_i) & \defEq
    \sup\Biggl(\big\{\!
    \begin{pmatrix}
      0\\
      \textcolor{conjColor}{1\,+\,} \sup_{i \in I} \expr_2(\psi_i)\\
      \textcolor{mainPosColor}{\sup_{i \in \mathit{Pos}}}\, \expr_1(\psi_i)\\
      \textcolor{otherPosColor}{\sup_{i \in \mathit{Pos} \setminus \mathit{R}}}\, \expr_1(\psi_i)\\
      \textcolor{negObsColor}{\sup_{i \in \mathit{Neg}}}\, \expr_1(\psi_i)\\
      0
    \end{pmatrix} \!\big\} \cup \{ \expr(\psi_i) \mid i \in I \}
    \Biggr)
    \\
    &
    \begin{array}{rl}
      \mathit{Neg} & \defEq \{i \in I \mid \exists \varphi'_i \ldotp \psi_i = \hmlNeg \varphi'_i \}\\
      \mathit{Pos} & \defEq I \setminus \mathit{Neg}\\
      \mathit{R} & \defEq
      \begin{cases}
        \varnothing & \text{if } \mathit{Pos} = \varnothing\\
        \{ r \} & \text{for some $r \in \mathit{Pos}$ where $\expr_1(\psi_r)$ maximal for $\mathit{Pos}$.}
      \end{cases}
    \end{array}
  \end{align*}
\end{definition}

\noindent
\refFig{fig:formula-price} gives an example how the prices compound.
The colors of the lines match those used for the dimensions and their updates in the other figures.
Circles mark the points that are counted.
The formula itself expresses a so-called \emph{ready-trace} observation:
We observe a trace $\tau \cdot \mathit{ec_A} \cdot \mathit{lc_A}$ and, along the way, may check what other options would have been enabled or disabled.
Here, we check that $\tau$ is enabled and $\mathit{ec_B}$ is disabled after the first $\tau$-step.
With the pricing, we can characterize all standard notions of equivalence:

\begin{proposition}
  \label{prop:language-prices}
  On finite systems,
  the languages of formulas with prices below the coordinates given in \refFig{fig:ltbt-spectrum} characterize the named notions of equivalence,
  that is, $p \bPreord{\mathit{X}} q$ with respect to equivalence $X$, iff no $\varphi$ with $\expr(\varphi) \leq e_X$ distinguishes $p$ from $q$.
\end{proposition}

\begin{figure}[t]
  \centering
  \vspace*{-5mm}
  \begin{tikzpicture}[]
    \begin{scope}[rotate=45, scale=1]
      \draw[help lines] (0,0) grid (4.5,4.5);
      \draw[thick, opacity=0.5, text opacity=.8] (-.25,2) -- (5,2) node[pos=0,left] {conjunctions = \textcolor{conjColor}{2}};
      \draw[thick, opacity=0.5, text opacity=.8] (1,-.25) -- (1,5) node[pos=0,right] {neg.\ clause depth = \textcolor{negObsColor}{1}};
      \draw[draw=none, fill=pink, fill opacity=0.1, text opacity=1]
        (-.05,-.05) rectangle (5,5) node[pos=.74, align=center] {bisimulation};
      \draw[draw=negObsColor, fill=negObsColor, opacity=0.7, fill opacity=0.2, text opacity=1]
        (-.05,-.05) rectangle (5,2.05) node[pos=.73, align=center] {(im)possible\\future};
      \draw[draw=orange, fill=orange, opacity=0.4, fill opacity=0.2, text opacity=1]
          (-.05,-.05) rectangle (1.05,5) node[pos=.73, align=center] {ready\\simulation};
      \draw[draw=green, fill=green, opacity=0.7, fill opacity=0.2, text opacity=1]
        (-.05,-.05) rectangle (1.05,2.05) node[below, align=center] {failure\\/ready};
      \draw[draw=yellow, fill=yellow, opacity=0.7, fill opacity=0.3, text opacity=1]
        (-.05,-.05) rectangle (.05,5) node[pos=.65] {simulation};
      \draw[draw=obsColor, fill=obsColor, opacity=0.7, fill opacity=0.2, text opacity=1]
        (-.05,-.05) rectangle (.05,1) node[below, align=center] {enabledness\\/trace};
      \draw[draw=red, fill=none, opacity=0.5, text opacity=1, line width=2pt]
        (1.95,3.5) -- (1.95, 2.95) node[above, align=center] {$\varphi_\ccsIdentifier{S'}$} -- (2.5,2.95);
      \draw[draw=red, fill=none, opacity=0.5, text opacity=1, line width=2pt]
        (0.95,2.5) -- (0.95, 1.95) node[above, align=center] {$\varphi_\ccsIdentifier{S}$} -- (1.5,1.95);
    \end{scope}
    \begin{scope}[transparency group]
      \fill[transform canvas={rotate=45, scale=1},path fading=west, color=white]
      (4,-.25) rectangle (5.05,5.05);
      \fill[transform canvas={rotate=45, scale=1},path fading=south, color=white]
      (-.25,4) rectangle (5.01,5.01);
    \end{scope}
  \end{tikzpicture}
  \vspace*{-5mm}
  \caption{Cut through the price lattice with dimensions 2 (conjunction nesting) and 5~(negated observation depth).}
  \label{fig:price-lattice}
\end{figure}

\begin{example}
  \label{exa:distinction-prices}
  The formulas of \refExample{exa:distinguishing-formulas} have prices:
  $\expr(\hmlObs{\tau}\hmlAndS\{\hmlNeg\hmlObs{ec_A}\}) = (2,2,0,0,1,1)$ for $\varphi_\ccsIdentifier{S}$
  and
  $\expr(\hmlAndS\{ \hmlNeg \hmlObs{\tau}\hmlAndS\{ \hmlNeg \hmlObs{ec_A} \} \}) = (2,3,0,0,2,2)$ for $\varphi_\ccsIdentifier{S'}$.
  The prices of the two are depicted as red marks in \refFig{fig:price-lattice}.
  This also visualizes how $\varphi_\ccsIdentifier{S'}$ is a counterexample for bisimilarity
  and how $\varphi_\ccsIdentifier{S}$ is a counterexample for failure and finer preorders.
  Indeed the two preorders are coarsest ways of telling the processes apart.
  So, $\ccsIdentifier{S}$ and $\ccsIdentifier{S'}$ are equated by all preorders \emph{below} the marks, \ie similarity, $\ccsIdentifier{S} \bEquiv{1S} \ccsIdentifier{S'}$, and coarser preorders ($\ccsIdentifier{S} \bEquiv{T} \ccsIdentifier{S'}$, $\ccsIdentifier{S} \bEquiv{E} \ccsIdentifier{S'}$).
  This carries over to the initial example of Peterson's mutex protocol from \refFig{fig:ts-peterson}, where weak simulation, $\ccsIdentifier{Pe} \bEquiv{1WS} \ccsIdentifier{Mx}$, is the most precise equivalence.
  Practically, this means that the specification $\ccsIdentifier{Mx}$ has liveness properties not upheld by the implementation $\ccsIdentifier{Px}$.
\end{example}

\begin{remark}
  \refDef{def:formula-prices} deviates from our previous formula pricing of~\cite{bjn2022decidingAllBehavioralEqs} in a crucial way:
  We only collect the \emph{maximal depths of positive clauses}, whereas~\cite{bjn2022decidingAllBehavioralEqs} tracks \emph{numbers of deep and flat positive clauses} (where a flat clause is characterized by an observation depth of $1$).
  Our change to a purely ``depth-guided'' spectrum will allow us to characterize the spectrum by an energy game and to eliminate the Bell-numbered blow up from the game's branching-degree.
  The special treatment of the deepest positive branch is necessary to address revival, failure trace, and ready trace semantics, which are popular in the CSP community~\cite{roscoe2009revivalsHierarchy,fhrr2004stuckfreeConformance}.
\end{remark}

\section{An Energy Game of Distinguishing Capabilities}
\label{sec:priced-game}

Conventional equivalence problems ask whether a pair of processes is related by a specific equivalence.
These problems can be abstracted into a more general ``spectroscopy problem'' to determine the set of equivalences from a spectrum that relate two processes as in~\cite{bjn2022decidingAllBehavioralEqs}.
This section captures the spectrum of \refFig{fig:ltbt-spectrum} by one rather simple energy game.

\subsection{Energy Games}
\label{subsec:energy-games}

Multidimensional energy games are played on graphs labeled by vectors to be added to (or subtracted from) a vector of ``energies'' where one player must pay attention to the energies not being exhausted.
We plan to encode the distinction capabilities of the semantic spectrum as energy levels in an energy game enriched by $\updMin{\dots}$-operations that takes minima of components.
This way, energy levels where the defender has a winning strategy will correspond to equivalences that hold.
We will just need updates decrementing or maintaining energy levels.

\begin{definition}[Energy updates]
  \label{def:updates}
  The set of \emph{energy updates}, $\energyUpdates$, contains vectors $\vectorComponents[N]{u} \in \energyUpdates$ where each component is of the form
  \begin{itemize}
    \item $u_k \in \{-1, 0\}$, or
    \item $u_k = \mathtt{min}_D$ where $D \subseteq \{1, \ldots, N\}$ and $k \in D$.
  \end{itemize}
  Applying an update to an energy, $\energyUpdate(e, u)$, where $e = \vectorComponents[N]{e} \in \energies$ (or $\energies_\infty$) and $u = \vectorComponents[N]{u} \in \energyUpdates$, yields a new energy vector $e'$ where $k$th components $e'_k \defEq e_k + u_k$ for $u_k \in \ints$ and $e'_k \defEq \min_{d\in D} e_d$ for $u_k = \mathtt{min}_D$. Updates that would cause any component to become negative are illegal.
\end{definition}

\begin{definition}[Games]
  \label{def:energy-game}
  An $N$-dimensional \emph{declining energy game} $\game[g_{0},e_0] = (G, G_\defenderSubscript, \gameMove, w, g_0, e_0)$ is played on a directed graph uniquely labeled by energy updates consisting of
  \begin{itemize}
    \item a set of \emph{game positions} $G$, partitioned into
    \begin{itemize}
      \item a set of \emph{defender positions} $G_\defenderSubscript\subseteq G$
      \item a set of \emph{attacker positions} $G_\attackerSubscript\defEq G \setminus G_\defenderSubscript$,
    \end{itemize}
    \item a relation of \emph{game moves} $\operatorname{\gameMove}\subseteq G \times G$,
    \item a \emph{weight function} for the moves $w \colon (\operatorname{\gameMove}) \to \energyUpdates$,
    \item an \emph{initial position} $g_{0}\in G$, and
    \item an \emph{initial energy budget} vector $e_{0} \in \energies_\infty$.
  \end{itemize}
  \noindent
  The notation $g \gameMoveX{u} g'$ stands for $g \gameMove g'$ and $w(g,g') = u$.
\end{definition}

\begin{definition}[Plays, energies, and wins]
  \label{def:plays-wins}
  We call the (finite or infinite) paths $\rho = g_{0}g_{1}\ldots\in G^{\infty}$ with $g_{i}\gameMoveX{u_i} g_{i+1}$ \emph{plays} of $\game[g_{0},e_{0}]$.

  The \emph{energy level} of a play $\rho$ at round $i$, $\energyLevel_{\rho}(i)$, is recursively defined as
  $\energyLevel_{\rho}(0) \defEq e_0$ and otherwise as
  $\energyLevel_{\rho}(i+1) \defEq \energyUpdate(\energyLevel_{\rho}(i), u_i)$.
  If we omit the index, $\energyLevel_{\rho}$, this refers to the final energy level of a finite run $\rho$, \ie, $\energyLevel_{\rho}(\relSize{\rho} - 1)$.

  Plays where energy levels become undefined (negative) are won by the defender.
  So are infinite plays.
  If a finite play is stuck (\ie, $g_{0}\dots g_{n} \centernot{\gameMove}$), the stuck player loses: The defender wins if $g_{n}\in G_\attackerSubscript$, and the attacker wins if $g_{n}\in G_\defenderSubscript$.
\end{definition}

\begin{proposition}
  \label{prop:declining-energies}
  In this model, energy levels can only decline.
  \begin{enumerate}
    \item Updates may only decrease energies, $\energyUpdate(e, u) \leq e$.
    \item Energy level changes are monotonic:
      If $\energyLevel_{\rho g} \leq \energyLevel_{\sigma g}$ and $g\gameMove g'$
      then $\energyLevel_{\rho g g'} \leq \energyLevel_{\sigma g g'}$.
    \item If $e_0 \leq e'_0$ and $\game[g_0, e_0]$ has non-negative play $\rho$, then $\game[g_0, e'_0]$ also has non-negative play $\rho$.
  \end{enumerate}
\end{proposition}

\begin{definition}[Strategies and winning budgets]
  \label{def:strategies}
  An \emph{attacker strategy} is a map from play prefixes ending in attacker positions to next game moves
  $s \colon (G^* \times G_\attackerSubscript) \to G$ with $s(g_0 \ldots g_a) \in (g_a \gameMove \cdot)$.
  Similarly, a \emph{defender strategy} names moves starting in defender states.
  If all plays consistent with a strategy $s$ ensure a player to win,
  $s$ is called a \emph{winning strategy} for this player.
  The player with a winning strategy for $\game[g_{0},e_0]$ is said to \emph{win} $\game$ from position $g_{0}$ with initial energy budget $e_0$.
  We call $\attackerWin(g) = \{e \mid \game[g, e] \text{ is won by the attacker}\}$ the \emph{attacker winning budgets}.
\end{definition}

\begin{proposition}
  \label{prop:up-win-budgets}
  The attacker winning budgets at positions are upward-closed with respect to energy, that is,
  $e \in \attackerWin(g)$ and $e \leq e'$ implies $e' \in \attackerWin(g)$.
\end{proposition}

\noindent
This means we can characterize the set of winning attacker budgets in terms of minimal winning budgets $\attackerWinMin(g) = \Min(\attackerWin(g))$, where $\Min(S)$ selects minimal elements $\{ e \in S \mid \nexists e' \in S \ldotp e' \leq e \land e' \neq e \}$.
Clearly, $\attackerWinMin$ must be an antichain and thus finite due to the energies being well-partially-ordered~\cite{kruskal1972wqo}.
Dually, we may consider the \emph{maximal} energy levels winning for the defender, $\defenderWinMax \colon G \to \powerSet{\energies_\infty}$ where we need extended energies to bound won half-spaces.

\subsection{The Spectroscopy Energy Game}
\label{subsec:spectroscopy-game}

Let us now look at the ``spectroscopy energy game'' at the center of our contribution.
Figure~\ref{fig:spec-game} gives a graphical representation.
The intuition is that the attacker shows how to construct formulas that distinguish a process $p$ from every $q$ in a set of processes $Q$.
The energies limit the expressiveness of the formulas.
The first dimension bounds for how many turns the attacker may challenge observations of actions.
The second dimension limits how often they may use conjunctions to resolve nondeterminism.
The third, fourth, and fifth dimensions limit how deeply observations may nest underneath a conjunction, where the fifth stands for negated clauses, the third for one of the deepest positive clauses and the fourth for the other positive clauses.
The last dimension limits how often the attacker may use negations to enforce symmetry by swapping sides.
The moves closely match productions in the grammar of \refDef{def:hml} and prices in \refDef{def:formula-prices}.

\begin{definition}[Spectroscopy energy game]
  \label{def:spectroscopy-game}
  For a system $\system=(\proc,\actions,\step{})$,
  the $6$-dimensional \emph{spectroscopy energy game} $\gameSpectroscopy^\system[g_0,e_0]=(G,G_\defenderSubscript,\gameMove,w,g_{0},e_0)$
  consists of
  \begin{itemize}
    \item \emph{attacker positions} $\attackerPos{p,Q} \in G_\attackerSubscript$,
    \item \emph{attacker clause positions} $\attackerPos[\land]{p,q} \in G_\attackerSubscript$,
    \item \emph{defender conjunction positions} $\defenderPos{p,Q,Q_*} \in G_\defenderSubscript$,
  \end{itemize}
  where $p, q \in \proc$ and $Q, Q_* \in \powerSet{\proc}$\!, and six kinds of moves:

  \medskip \noindent
  \begin{tabularx}{\linewidth}{@{}l@{\;}l@{\;}c@{\;}l@{\;}X@{}}
    \hspace{\labelsep}\textminus\hspace{\labelsep}\emph{observation moves}
    & $\attackerPos{p,Q}$
    & $\gameMoveX{(-1,0,0,0,0,0)}$
    & $\mathmakebox[4em][l]{\attackerPos{p^\prime, Q'}}$
      if $p \step{\ccsInm{a}} p^\prime$\!, $Q \step{\ccsInm{a}} Q'$\!,
    \\
    \hspace{\labelsep}\textminus\hspace{\labelsep}\emph{conj.~challenges}
    & $\attackerPos{p,Q}$
    & $\gameMoveX{(0,-1,0,0,0,0)}$
    & $\defenderPos{p,Q \setminus Q_*, Q_*}$
      \; if $Q_* \subseteq Q$,
    \\
    \hspace{\labelsep}\textminus\hspace{\labelsep}\emph{conj.~revivals}
    & $\defenderPos{p,Q,Q_*}$
    & $\gameMoveX{(\updMin{1,3},0,0,0,0,0)}$
    & $\mathmakebox[4em][l]{\attackerPos{p,Q_*}}$
      if $Q_* \neq \varnothing$,
    \\
    \hspace{\labelsep}\textminus\hspace{\labelsep}\emph{conj.~answers}
    & $\defenderPos{p,Q,Q_*}$
    & $\gameMoveX{(0,0,0,\updMin{3,4},0,0)}$
    & $\mathmakebox[4em][l]{\attackerPos[\land]{p,q}}$
      if $q \in Q$,
    \\
    \hspace{\labelsep}\textminus\hspace{\labelsep}\emph{positive decisions}
    & $\attackerPos[\land]{p,q}$
    & $\gameMoveX{(\updMin{1,4},0,0,0,0,0)}$
    & $\mathmakebox[4em][l]{\attackerPos{p,\{q\}},}$ and
    \\
    \hspace{\labelsep}\textminus\hspace{\labelsep}\emph{negative decisions}
    & $\attackerPos[\land]{p,q}$
    & $\gameMoveX{(\updMin{1,5},0,0,0,0,-1)}$
    & $\mathmakebox[4em][l]{\attackerPos{q,\{p\}}}$
      if $p \neq q$.
    \\
  \end{tabularx}
\end{definition}

\begin{figure}[t]
  \centering
  \begin{tikzpicture}[>->,shorten <=1pt,shorten >=0.5pt,auto,node distance=2cm, rel/.style={dashed,font=\it},
    posStyle/.style={draw, inner sep=1ex,minimum size=1cm,minimum width=2cm,anchor=center,draw,black,fill=gray!5}]
      \node[posStyle]
        (Att){$\attackerPos{p,Q}$};
      \node[ellipse, draw, inner sep=1ex, minimum size=1cm,minimum width=2cm,fill=gray!5]
        (Def) [right = 2.2cm of Att] {$\defenderPos{p,Q \setminus Q_*,Q_*}$};
      \node[posStyle]
        (AttConj) [below right = 1.2cm and 3cm of Att] {$\attackerPos[\land]{p,q}$};
      \node[posStyle, dashed]
        (AttSwap) [right = 3cm of AttConj] {$\attackerPos{q,\{p\}}$};
      \node[posStyle, dashed]
        (AttContinue) [above = 1cm of AttSwap] {$\attackerPos{p,\{q\}}$};
      \node[posStyle, dashed]
        (AttObs) [above = 1cm of AttContinue] {$\attackerPos{p^\prime,Q^\prime}$};

      \path
        (Att) edge[bend left=28]
          node[pos=.3, align=center, label={[label distance=0.0cm]-25:$\textcolor{gray}{\textcolor{obsColor}{-1},0,0,0,0,0}$}] {$p\step{a}p^\prime$\\ $Q\step{a}Q^\prime$} (AttObs)
        (Att) edge
          node[label={[label distance=0cm]-90:$\textcolor{gray}{0,\textcolor{conjColor}{-1},0,0,0,0}$}] {$Q_* \subseteq Q$} (Def)
        (Def) edge[bend left=15]
          node[align=right, swap, pos=.1] {$q \in Q \setminus Q_*$\\$\textcolor{gray}{0,0,0,\updMin{\textcolor{mainPosColor}{3},\textcolor{otherPosColor}{4}},0,0}$} (AttConj)
        (AttConj) edge[bend left=10]
          node[label={[label distance=0.0cm]-75:$\textcolor{gray}{\updMin{\textcolor{obsColor}{1},\textcolor{otherPosColor}{4}},0,0,0,0,0}$}] {} (AttContinue)
        (AttConj) edge[bend right=15]
          node[align=center, label={[label distance=0.0cm]-90:$\textcolor{gray}{\updMin{\textcolor{obsColor}{1},\textcolor{negObsColor}{5}},0,0,0,0,\textcolor{negsColor}{-1}}$}] {} (AttSwap)
        (Def) edge[bend left=10]
          node[align=center, label={[label distance=0.0cm]-30:$\textcolor{gray}{\updMin{\textcolor{obsColor}{1},\textcolor{mainPosColor}{3}},0,0,0,0,0}$}] {$p' = p$\\$Q'=Q_* \neq \varnothing$} (AttObs);
  \end{tikzpicture}
  \caption{
    Schematic spectroscopy game $\gameSpectroscopy$ of \refDef{def:spectroscopy-game}.}
  \label{fig:spec-game}
\end{figure}

\noindent
The spectroscopy energy game is a bisimulation game in the tradition of Stirling~\cite{stirling1993modal}.

\begin{lemma}[Bisimulation game, proof see~%
  \ifthenelse{\boolean{fullversion}}
    {page~\pageref{prf:bisimulation-game}}
    {\cite{bisping2023equivalenceEnergyGamesReport}}]
  \label{lem:bisimulation-game}
  $p_0$ and $q_0$ are bisimilar
  iff the defender wins $\gameSpectroscopy[\attackerPos{p_0, \{q_0\}}, e_0]$ for every initial energy budget $e_0$,
  \ie if $(\infty, \infty, \infty, \infty, \infty, \infty) \in \defenderWinMax(\attackerPos{p_0, \{q_0\}})$.
\end{lemma}
In other words, if there are initial budgets winning for the attacker, then the compared processes can be told apart.
For $\gameSpectroscopy$, the attacker ``unknown initial credit problem'' in energy games~\cite{vcdh2015complexityMeanPayoffEnergy} coincides with the ``apartness problem''~\cite{geuversJacobs2021apartness} for processes.

\newcommand*{\seprule}{\\[-5pt]\rule{1.3cm}{.5pt}\\[-1pt]}

 \begin{figure}[t!]
  \centerfloat
  \begin{tikzpicture}[auto,shorten <=1pt,shorten >=0.5pt,
    defender/.style={ellipse, inner sep=0ex},
    defenderWins/.style={draw=blue},
    position/.style={inner sep=5pt,align=center,anchor=center,draw,black,fill=gray!5,thick, node font=\small}]
    \begin{scope}[]
      \node[position, initial, initial text={}, label={above left:$\color{magenta}\hmlObs{\tau}\hmlAndS\{\hmlNeg\hmlObs{ec_A}\}$}] (s0s1) at(0,10) {
          $\attackerPos{\ccsIdentifier{S}, \{\ccsIdentifier{S'}\}}$
        \seprule
        (2,2,0,0,1,1)
      };
      \node[position] (s0s1c) at(5,11) {
        $\attackerPos[\land]{\ccsIdentifier{S}, \ccsIdentifier{S'}}$
        \seprule
        (2,2,0,2,1,1)\\(2,3,0,0,2,3)
      };
      \node[position, defender] (s0s1d) [above left = -.4cm and 1.3cm of s0s1c] {
        $\defenderPos{\ccsIdentifier{S}, \{\ccsIdentifier{S'}\}, \varnothing}$
        \seprule
        (2,2,2,2,1,1)\\(2,3,0,0,2,3)
      };
      \node[position, defenderWins] (DivDiv) [below right = .5cm and .5cm of s0s1] {
        $\attackerPos{\ccsIdentifier{Div}, \{ \ccsIdentifier{Div} \}}$
      };
      \node[position,label={[xshift=-.5cm]above right:$\color{magenta}\hmlAndS\{\hmlNeg\hmlObs{\tau}\hmlAndS\{\hmlNeg\hmlObs{ec_A}\}\}$}] (s1s0) at(10,10) {
        $\attackerPos{\ccsIdentifier{S'}, \{\ccsIdentifier{S}\}}$
        \seprule
        (2,3,0,0,2,2)
      };
      \node[position, label={above:$\color{magenta}\hmlNeg\hmlObs{\tau}\hmlAndS\{\hmlNeg\hmlObs{ec_A}\}$}] (s1s0c) [below = .8cm of s0s1c] {
        $\attackerPos[\land]{\ccsIdentifier{S'}, \ccsIdentifier{S}}$
        \seprule
        (2,2,0,0,2,2)
      };
      \node[position, defender, label={[xshift=1.2cm]above left:$\color{magenta}\hmlAndS\{\hmlNeg\hmlObs{\tau}\hmlAndS\{\hmlNeg\hmlObs{ec_A}\}\}$}] (s1s0d) [below left = 1cm and -.5cm of s1s0]  {
        $\defenderPos{\ccsIdentifier{S'}, \{\ccsIdentifier{S}\}, \varnothing}$
        \seprule
        (2,2,0,0,2,2)
      };
      \node[position] (s1s0Div) [below = 1cm of s1s0d]  {
        $\attackerPos{\ccsIdentifier{S'}, \{\ccsIdentifier{S}, \ccsIdentifier{Div}\}}$
        \seprule
        (2,3,0,0,2,2)
      };
      \node[position, defender] (s1Divds0) [below right = .1cm and 1.2cm of s1s0Div] {
        $\defenderPos{\ccsIdentifier{S'}, \{\ccsIdentifier{Div}\}, \{ \ccsIdentifier{S} \}}$
        \seprule
        (2,3,2,0,2,2)
      };
      \node[position, defender] (s1Divs0d) [below = 1.2cm of s1s0Div] {
        $\defenderPos{\ccsIdentifier{S'}, \{\ccsIdentifier{S}, \ccsIdentifier{Div}\}, \varnothing}$
        \seprule
        (2,2,0,0,2,2)
      };
      \node[position, defenderWins, dashed] (DivDiv2) [above right = .5cm and .5cm of s1s0Div] {
        $\attackerPos{\ccsIdentifier{Div}, \{ \ccsIdentifier{Div} \}}$
      };
      \node[position, label={above left:$\color{magenta}\hmlAndS\{\hmlNeg\hmlObs{ec_A}\}$}] (DivS1) [below = 1cm of s0s1] {
        $\attackerPos{\ccsIdentifier{Div}, \{ \ccsIdentifier{S'} \}}$
        \seprule
        (1,2,0,0,1,1)
      };
      \node[position, defender, label={above left:$\color{magenta}\hmlAndS\{\hmlNeg\hmlObs{ec_A}\}$}] (DivS1d) [below = 1cm of DivS1] {
        $\defenderPos{\ccsIdentifier{Div}, \{ \ccsIdentifier{S'} \}, \varnothing}$
        \seprule
        (1,1,0,0,1,1)
      };
      \node[position, label={above left:$\color{magenta}\hmlNeg\hmlObs{ec_A}$}] (DivS1c) [right = 1cm of DivS1d] {
        $\attackerPos[\land]{\ccsIdentifier{Div}, \ccsIdentifier{S'}}$
        \seprule
        (1,1,0,0,1,1)
      };
      \node[position, label={above left:$\color{magenta}\hmlObs{ec_A}$}] (S1Div) [below = 1cm of DivS1c] {
        $\attackerPos{\ccsIdentifier{S'}, \{ \ccsIdentifier{Div} \}}$
        \seprule
        (1,1,0,0,0,0)
      };
      \node[position, defender] (S1Divd) [below = 1cm of S1Div] {
        $\defenderPos{\ccsIdentifier{S'}, \{ \ccsIdentifier{Div} \}, \varnothing}$
        \seprule
        (1,1,1,1,0,0)\\(1,2,0,0,1,2)
      };
      \node[position] (S1Divc) [below right = 1cm and 1.8cm of S1Div] {
        $\attackerPos[\land]{\ccsIdentifier{S'}, \ccsIdentifier{Div}}$
        \seprule
        (1,1,0,1,0,0)\\(1,2,0,0,1,2)
      };
      \node[position, label={above left:$\color{magenta}\hmlTrue\vphantom{\hmlObs{ec_A}}$}] (DivEmp) [left = 1.2cm of S1Div] {
        $\attackerPos{\ccsIdentifier{Div}, \varnothing}$
        \seprule
        (0,1,0,0,0,0)
      };
      \node[position, defender, label={above left:$\color{magenta}\hmlTrue$}] (DivEmpd) [below = 1cm of DivEmp] {
        $\defenderPos{\ccsIdentifier{Div}, \varnothing, \varnothing}$
        \seprule
        (0,0,0,0,0,0)
      };
    \end{scope}
    \begin{scope}[>->,black!75,every node/.style={node font=\small}]
      \draw (s0s1) to[bend left=20] node {$0,-1$} (s0s1d);
      \draw (s0s1) to[out=230,in=200,looseness=3] node {$-1,0$} (s0s1);
      \draw (s0s1) to[bend left=20, swap] node {$-1,0$} (DivDiv);
      \draw[thick] (s0s1d) to[bend left=20] node {$0,0,0,\updMin{3,4}$} (s0s1c);
      \draw[thick] (s0s1c) to[bend right=15] node {$\updMin{1,4}$} (s0s1);
      \draw[thick] (s0s1c) to[bend left=20] node {$\updMin{1,5},0,0,0,0,-1$} (s1s0);
      \draw[thick] (s0s1) to node {$-1,0$} (DivS1);
      \draw[thick] (s1s0) to[bend left=20, swap] node {$0,-1$} (s1s0d);
      \draw (s1s0) to[pos=.35, bend left=22] node {$-1,0$} (s1s0Div);
      \draw (s1s0) to node {$-1,0$} (DivDiv2);
      \draw[thick] (s1s0d) to[pos=.3,bend left=30] node {$0,0,0,\updMin{3,4}$} (s1s0c);
      \draw (s1s0c) to[bend left=10] node {$\updMin{1,4}$} (s1s0);
      \draw[thick] (s1s0c) to[bend right=10, swap] node {$\updMin{1,5},0,0,0,0,-1$} (s0s1);
      \draw (s1s0Div) to[bend left=20] node {$0,-1$} (s1Divds0);
      \draw (s1s0Div) to[pos=.7, swap] node {$-1,0$} (DivDiv2);
      \draw (s1s0Div) to[out=230,in=200,looseness=3,pos=.3] node {$-1,0$} (s1s0Div);
      \draw[thick] (s1Divds0) to[bend right=45, swap] node {$\updMin{1,3}$} (s1s0);
      \draw[thick] (s1Divds0) to[bend left=40] node {$0,0,0,\updMin{3,4}$} (S1Divc);
      \draw[thick] (s1s0Div) to[bend left=20] node {$0,-1$} (s1Divs0d);
      \draw[thick] (s1Divs0d) to[bend left=15, swap, pos=.65] node {$0,0,0,\updMin{3,4}$} (s1s0c);
      \draw[thick] (s1Divs0d) to[bend left=20, pos=.35] node {$0,0,0,\updMin{3,4}$} (S1Divc);
      \draw[thick] (DivS1) to[swap] node {$0,-1$} (DivS1d);
      \draw (DivS1) to[out=230,in=200,looseness=3] node {$-1,0$} (DivS1);
      \draw[thick] (DivS1d) to[bend right=20, swap] node {$0,0,0,\updMin{3,4}$} (DivS1c);
      \draw (DivS1c.north) to[bend right=15] node[pos=.8] {$\updMin{1,4}$} (DivS1.south east);
      \draw[thick] (DivS1c) to[pos=.65] node {$\updMin{1,5},0,0,0,0,-1$} (S1Div);
      \draw (S1Div) to node {$0,-1$} (S1Divd);
      \draw[thick] (S1Div) to node {$-1,0$} (DivEmp);
      \draw (S1Div) to[out=230,in=200,looseness=3] node {$-1,0$} (S1Div);
      \draw[thick] (S1Divd) to[bend right=29, swap] node {$0,0,0,\updMin{3,4}$} (S1Divc);
      \draw[thick] (DivEmp) to[swap] node {$0,-1$} (DivEmpd);
      \draw[thick] (S1Divc) to[bend right=10] node[pos=.3, swap] {$\updMin{1,4}$} (S1Div);
      \draw[thick] ([xshift=-.1cm]S1Divc.north) to[bend right=40] node[pos=.89, swap] {$\updMin{1,5},0,0,0,0,-1$} (DivS1.east);
    \end{scope}
  \end{tikzpicture}
  \caption{\refExample{exa:example-game} spectroscopy energy game, minimal attacker winning budgets, and distinguishing formulas/clauses.
  (In order to reduce visual load, only the first components of the updates are written next to the edges. The other components are 0.)}
  \label{fig:example-game}
\end{figure}

\begin{example}
  \label{exa:example-game}
  \refFig{fig:example-game} shows the spectroscopy energy game starting at $\attackerPos{\ccsIdentifier{S},\{\ccsIdentifier{S'}\}}$ from \refExample{exa:distinguishing-formulas}.
  The lower part of each node displays the node's $\attackerWinMin$.
  The magenta HML formulas illustrate distinctions relevant for the correctness argument of the following \refSubsec{subsec:game-correctness}.
  \refSec{sec:computing-equivalences} will describe how to obtain attacker winning budgets and equivalences.
  The blue ``symmetric'' positions are definitely won by the defender---we omit the game graph below them.
  We also omit the move $\attackerPos{\ccsIdentifier{S'}, \{\ccsIdentifier{S}, \ccsIdentifier{Div}\}}\gameMoveX{(0,-1,0,0,0,0)}\defenderPos{\ccsIdentifier{S'}, \{\ccsIdentifier{S} \}, \{ \ccsIdentifier{Div}\}}$---it can be ignored as will be discussed in \refSubsec{subsec:more-clever}.
\end{example}

\subsection{Correctness: Tight Distinctions}
\label{subsec:game-correctness}

We will check that winning budgets indeed characterize what equivalences hold by constructing price-minimal distinguishing formulas from attacker budgets.

\begin{definition}[Strategy formulas]
  Given the set of winning budgets $\attackerWin$,
  the set of \emph{attacker strategy formulas} $\hmlStrategies$ for a position with given energy level $e$ is defined inductively as follows:
  \begin{description}[itemsep=4pt]
    \item $\hmlObs{b}\varphi \in \hmlStrategies(\attackerPos{p,Q}, e)$
      \; if \; 
      $\attackerPos{p,Q} \gameMoveX{u} \attackerPos{p',Q'}$,
      $e' = \energyUpdate(e, u) \!\in \attackerWin(\attackerPos{p',Q'})$,
      $p \step{b} p'$,
      $Q \step{b} Q'$,
      and $\varphi \in \hmlStrategies(\attackerPos{p',Q'}, e')$,
    \item $\varphi \!\in\! \hmlStrategies(\attackerPos{p,Q}, e)$
    \; if \; 
      $\attackerPos{p,Q} \!\gameMoveX{u} \defenderPos{p,Q,Q_*}$,
      $e'\! = \!\energyUpdate(e, u)\! \in\! \attackerWin(\defenderPos{p,Q,Q_*})$,
      and $\varphi \in \hmlStrategies(\defenderPos{p,Q,Q_*}, e')$,
    \item $\hmlAnd{q}{Q}\psi_q \in \hmlStrategies(\defenderPos{p,Q,\varnothing}, e)$
      \; if \; 
      $\defenderPos{p,Q,\varnothing} \!\gameMoveX{u_q}\! \attackerPos[\land]{p,q}$,
      $e_q \!=\! \energyUpdate(e, u_q) \!\in\! \attackerWin(\attackerPos[\land]{p,q})$
      and $\psi_q \in \hmlStrategies(\attackerPos[\land]{p,q}, e_q)$ for each $q \in Q$,
    \item $\hmlAnd{q}{Q \cup \{*\}}\psi_q \in \hmlStrategies(\defenderPos{p,Q,Q_*}, e)$
      \; if \; 
      $\defenderPos{p,Q,Q_*} \!\gameMoveX{u_q}\! \attackerPos[\land]{p,q}$,
      $e_q \!=\! \energyUpdate(e, u_q) \!\in\! \attackerWin(\attackerPos[\land]{p,q})$
      and $\psi_q \in \hmlStrategies(\attackerPos[\land]{p,q}, e_q)$ for each $q \in Q$,
      and if
      $\defenderPos{p,Q,Q_*} \!\gameMoveX{u_*}\! \attackerPos{p,Q_*}$,
      $e_* \!=\! \energyUpdate(e, u_*) \!\in\! \attackerWin(\attackerPos{p,Q_*})$, and
      $\psi_* \in \hmlStrategies(\attackerPos{p,Q_*}, e_*)$ is an observation,
    \item $\varphi \in \hmlStrategies(\attackerPos[\land]{p,q}, e)$
      \; if \; 
      $\attackerPos[\land]{p,q} \gameMoveX{u} \attackerPos{p,\{q\}}$,
      $e' = \energyUpdate(e, u) \in \attackerWin(\attackerPos{p,\{q\}})$
      and $\varphi \in \hmlStrategies(\attackerPos{p,\{q\}}, e')$ is an observation, and
    \item $\hmlNeg \varphi \in \hmlStrategies(\attackerPos[\land]{p,q}, e)$
      \; if \; 
      $\attackerPos[\land]{p,q} \gameMoveX{u} \attackerPos{q,\{p\}}$,
      $e' = \energyUpdate(e, u) \in \attackerWin(\attackerPos{q,\{p\}})$
      and $\varphi \in \hmlStrategies(\attackerPos{q,\{p\}}, e')$ is an observation.
  \end{description}
\end{definition}
Because of the game structure, we actually know the $u$ needed in each line of the definition.
It is $u = (-1,0,0,0,0,0)$ in the first case; $(0,-1,0,0,0,0)$ in the second; $(0,0,0,\updMin{3,4},0,0)$ in the third; $(0,0,0,\updMin{3,4},0,0)$ and $(\updMin{1,3},0,0,0,0,0)$ in the fourth; $(\updMin{1,4},0,0,0,0,0)$ in the fifth; and $(\updMin{1,5},0,0,0,0,-1)$ in last case.
$\hmlStrategies(\attackerPos[\land]{p,q}, \cdot)$ can contain negative clauses, which form no proper formulas on their own.

\begin{lemma}[Price soundness]
  \label{lem:price-soundness}
  $\varphi \in \hmlStrategies(\attackerPos{p,Q}, e)$ implies that $\expr(\varphi) \leq e$ and that $\expr(\varphi) \in \attackerWin(\attackerPos{p,Q})$.
\end{lemma}
\begin{proof}
  By induction on the structure of $\varphi$ with arbitrary $p,Q,e$,
  exploiting the alignment of the definitions of winning budgets and formula prices.
  Full proof in~%
  \ifthenelse{\boolean{fullversion}}
    {the appendix on page~\pageref{prf:price-soundness}}
    {\cite{bisping2023equivalenceEnergyGamesReport}}.
\end{proof}

\begin{lemma}[Price completeness]
  \label{lem:price-completeness}
  $e_0 \in \attackerWin(\attackerPos{p_0,Q_0})$ implies there are elements in $\hmlStrategies(\attackerPos{p_0,Q_0}, e_0)$.
\end{lemma}
\begin{proof}
  By induction on the tree of winning plays consistent with some attacker winning strategy
  implied by $e_0 \in \attackerWin(\attackerPos{p_0,Q_0})$.
  Full proof in~%
  \ifthenelse{\boolean{fullversion}}
    {the appendix on page~\pageref{prf:price-completeness}}
    {\cite{bisping2023equivalenceEnergyGamesReport}}.
\end{proof}

\begin{lemma}[Distinction soundness]
  \label{lem:distinction-soundness}
  Every $\varphi \in \hmlStrategies(\attackerPos{p,Q}, e)$ distinguishes $p$ from every $q \in Q$.
\end{lemma}
\begin{proof}
  By induction on the structure of $\varphi$ with arbitrary $p,Q,e$,
  exploiting that $\hmlStrategies$ can only construct formulas with the invariant that they are true for $p$ and false for each $q \in Q$.
  Full proof in~%
  \ifthenelse{\boolean{fullversion}}
    {the appendix on page~\pageref{prf:distinction-soundness}}
    {\cite{bisping2023equivalenceEnergyGamesReport}}.
\end{proof}

\begin{lemma}[Distinction completeness]
  \label{lem:distinction-completeness}
  If $\varphi$ distinguishes $p$ from every $q \in Q$, then $\expr(\varphi) \in \attackerWin(\attackerPos{p, Q})$.
\end{lemma}
\begin{proof}
  By induction on the structure of $\varphi$ with arbitrary $p,Q$,
  exploiting the alignment of game structure and HML semantics
  and the fact that $\expr$ cannot ``overtake'' inverse updates.
  Full proof in~%
  \ifthenelse{\boolean{fullversion}}
    {the appendix on page~\pageref{prf:distinction-completeness}}
    {\cite{bisping2023equivalenceEnergyGamesReport}}.
\end{proof}

\begin{theorem}[Correctness]
  \label{thm:correctness}
  For any equivalence $X$ with coordinate $e_X$,
  $p \bPreord{X} q$, precisely if
  all $e_{pq} \in \attackerWinMin(\attackerPos{p, \{q\}})$ are above or incomparable, $e_{pq} \not \leq e_X$.
\end{theorem}
\begin{proof}
  By contraposition, in both directions.
  \begin{itemize}
    \item Assume $p \not\bPreord{\mathit X} q$.
      This means some $\varphi$ with $\expr(\varphi) \leq e_X$ distinguishes $p$ from $q$.
      By \refLem{lem:distinction-completeness}, $\expr(\varphi) \in \attackerWin(\attackerPos{p, \{q\}})$.
      Then either $\expr(\varphi)$ or a lower price $e_{pq} \leq \expr(\varphi)$ are minimal winning budgets,
      \ie, $e_{pq} \in \attackerWinMin(\attackerPos{p, \{q\}})$ and $e_{pq} \leq e_X$.
    \item Assume there is $e_{pq} \in \attackerWinMin(\attackerPos{p, \{q\}})$ with $e_{pq} \leq e_X$.
      By \refLem{lem:price-completeness}, there is $\varphi \in \hmlStrategies(\attackerPos{p,\{q\}}, e_{pq})$.
      Due to \refLem{lem:distinction-soundness}, $\varphi$ must be distinguishing for $p$ and $q$,
      and due to \refLem{lem:price-soundness}, $\expr(\varphi) \leq e_{pq} \leq e_X$.
  \end{itemize}
\end{proof}

\noindent
The theorem basically means that by fixing an initial budget in $\gameSpectroscopy$, we can obtain a characteristic game for any notion from the spectrum.

\subsection{Becoming More Clever by Looking One Step Ahead}
\label{subsec:more-clever}

The spectroscopy energy game $\gameSpectroscopy$ of \refDef{def:spectroscopy-game} may branch exponentially with respect to $\relSize{Q}$ at conjunction challenges after $\attackerPos{p,Q}$.
For the spectrum we are interested in, we can drastically limit the sensible attacker moves to four options by a little lookahead into the enabled actions $\initials(q)$ of $q \in Q$ and $\initials(p)$.

\begin{definition}[Clever spectroscopy game]
  The \emph{clever spectroscopy game}, $\gameSpectroscopyClever$, is defined exactly like the previous spectroscopy energy game of \refDef{def:spectroscopy-game} with the restriction of the conjunction challenges
  \[
    \attackerPos{p,Q}
    \quad \gameMoveX{(0,-1,0,0,0,0)}_\blacktriangle
    \quad \defenderPos{p,Q \setminus Q_*, Q_*}
    \quad \text{with } Q_* \subseteq Q,
  \]
  to situations where
  $\; Q_* \in \{ \; \varnothing,\quad
               \{ q \in Q \mid \initials(q) \subseteq \initials(p) \},\quad
               \{ q \in Q \mid \initials(p) \subseteq \initials(q) \},\allowbreak\quad
               \{ q \in Q \mid \initials(p) = \initials(q) \} \; \}.$
\end{definition}

\begin{theorem}[Correctness of cleverness]
  \label{thm:correctness-clever}
  Assume modal depth of positive clauses $e_4 \in \{0, 1, \infty\}$, $e_4 \leq e_3$, and that modal depth of negative clauses $e_5 > 1$ implies $e_3 = e_4$. Then, the attacker wins $\gameSpectroscopyClever[\attackerPos{p_0,Q_0}, e]$ precisely if they win $\gameSpectroscopy[\attackerPos{p_0,Q_0}, e]$.
\end{theorem}
\begin{proof}
  The implication from the clever spectroscopy game $\gameSpectroscopyClever$ to the full spectroscopy game $\gameSpectroscopy$ is trivial as the attacker moves in ${\gameMove_\blacktriangle}$ are a subset of those in ${\gameMove_\triangle}$ and the defender has the same moves in both games.
  For the other direction, we have to show that any move $\attackerPos{p,Q} \gameMoveX{(0,-1,0,0,0,0)}_\triangle\defenderPos{p, Q \setminus Q_*, Q_*}$ winning at energy level $e$ can be simulated by a winning move $\attackerPos{p,Q} \gameMoveX{(0,-1,0,0,0,0)}_\blacktriangle\defenderPos{p, Q \setminus Q', Q'}$.
  Full proof in~%
  \ifthenelse{\boolean{fullversion}}
    {the appendix on page~\pageref{prf:correctness-clever}}
    {\cite{bisping2023equivalenceEnergyGamesReport}}.
\end{proof}

\section{Computing Equivalences}
\label{sec:computing-equivalences}

The previous section has shown that attacker winning budgets in the spectroscopy energy game characterize distinguishable processes and, dually, that the defender's wins characterize equivalences.
We now examine how to actually compute the winning budgets of both players.

\subsection{Computation of Attacker Winning Budgets}
\label{subsec:computation-winning}

The winning budgets of the attacker (\refDef{def:strategies}) are characterized inductively:

\begin{itemize}
  \item Where the defender is stuck, $g \in G_\defenderSubscript$ and $g \centernot{\gameMove}$,
  the attacker wins with any budget, $(0,0,0,0,0,0) \in \attackerWinMin(g)$.
  \item Where the defender has moves, $g \in G_\defenderSubscript$ and $g \gameMoveX{u_i} g'_i$ (for some indexing $i \in I$ over all possible moves),
   the attacker wins if they have a budget equal or above to all budgets that might be necessary after the defender's move:
   If $\energyUpdate(e, u_i) \in \attackerWin(g'_i)$ for all $i \in I$, then $e \in \attackerWin(g)$.
  \item Where the attacker moves, $g \in G_\attackerSubscript$ and $g \gameMoveX{u} g'$,
  $\energyUpdate(e, u) \in \attackerWin(g')$ implies $e \in \attackerWin(g)$.
\end{itemize}

\noindent
By \refProp{prop:up-win-budgets}, it suffices to find the finite set of minimal winning budgets, $\attackerWinMin$.
Turning this into a computation is not as straightforward as in other energy game models.
Due to the $\mathtt{min}_D$-updates, the energy update function $\energyUpdate(\cdot, u)$ is neither injective nor surjective.

We must choose an inversion function $\energyUpdateInv$ that picks minimal solutions and that minimally ``casts up'' inputs that are outside the image of $\energyUpdate(\cdot, u)$,
i.e., such that $\energyUpdateInv(e', u) = \inf \{ e \mid e' \leq \energyUpdate(e, u) \}$.
We compute it as follows:

\begin{definition}[Inverse update]
  \label{def:inverse-update}
  The \emph{inverse update} function is defined as
  $\energyUpdateInv(e', u) \mathrel{\defEq}
    \sup (\{ e \} \cup \{ m(i) \mid
      \exists D \ldotp u_i = \mathtt{min}_D \})$
  with $e_i = e'_i - u_i$ for all $i$ where $u_i \in \{0,-1\}$
    and $e_i = e'_i$ otherwise,
  and with $(m(i))_j = e'_i$ for $u_i = \mathtt{min}_D$ and $j \in D$,
    and $(m(i))_j = 0$ otherwise, for all $i, j$.
\end{definition}

\begin{example}
  Let $u \defEq (\updMin{1,3},\updMin{1,2},-1,-1)$. $(3,4,0,1) \notin \mathrm{img}(\energyUpdate(\cdot, u))$, but:
  \begin{align*}
    \smash{\energyUpdateInv((3,4,0,1), u)} & = \sup \{ (3,4,1,2), (3,0,3,0), (4,4,0,0) \} = (4,4,3,2) \\
    \smash{\energyUpdate((4,4,3,2), u)} & = (3,4,2,1) \geq (3,4,0,1)\\
    \smash{\energyUpdateInv((3,4,2,1), u)} & = \sup \{ (3,4,3,2), (3,0,3,0), (4,4,0,0) \} = (4,4,3,2)
  \end{align*}
\end{example}

\begin{algorithm}[t]
  \Fn{$\varname{compute\_winning\_budgets}(\game=(G,G_\defenderSubscript,\gameMove,w))$}{

    $\varname{attacker\_win}\, := [g\mapsto \{\} \mid g\in G]$

    $\varname{todo}\, := \{g \in G_\defenderSubscript \mid g \not \gameMoveX{\cdot} \}$

    \While{$\varname{todo} \neq \varnothing $}{

      $\varname{g} := \KwSty{some}\; \varname{todo}$

      $\varname{todo} := \varname{todo} \setminus \{ \varname{g} \}$

      \eIf{$\varname{g} \in G_\attackerSubscript$}{

        $\varname{new\_attacker\_win} :=
          \Min ( \varname{attacker\_win}[\varname{g}] \cup \{\energyUpdateInv(\mathit{e'}, u) \mid {\varname{g} \gameMoveX{u} g'} \land \mathit{e'} \in \varname{attacker\_win}[g']\})$
        \label{code:attacker-pos-update}

      } {

        $\varname{defender\_post} := \{g' \mid \varname{g} \gameMoveX{u} g'\}$

        $\varname{options} :=
          \{(g', \energyUpdateInv(\mathit{e'}, u)) \mid
            {\varname{g} \gameMoveX{u} g'} \land {\mathit{e'} \in \varname{attacker\_win}[g']}\}\}$

        \eIf{$\varname{defender\_post} \subseteq \domain(\varname{options})$ \label{code:option-covering}}{

          $\varname{new\_attacker\_win} :=
            \Min(
              \{\sup_{g' \in \varname{defender\_post}} \mathit{strat}(g') \mid
                {\mathit{strat} \in (G \rightarrow \energies)} \land
                \forall g'\! \ldotp \mathit{strat}(g') \in \varname{options}(g') \})$
          \label{code:defender-pos-update}

        } {

          $\varname{new\_attacker\_win} := \varnothing$

        }

      }

      \If{$\varname{new\_attacker\_win} \neq \varname{attacker\_win}[\varname{g}]$}{

        $\varname{attacker\_win}[\varname{g}] := \varname{new\_attacker\_win}$

        $\varname{todo} := \varname{todo} \cup \{ g_p \mid \exists u \ldotp g_p \gameMoveX{u} \varname{g} \}$

      }

    }

    $\attackerWinMin := \varname{attacker\_win}$

    \KwRet{$\attackerWinMin$}

  }
  \caption{\label{alg:game-algorithm}Algorithm for computing attacker winning budgets of declining energy game $\game$.}
\end{algorithm}

\noindent
With $\energyUpdateInv\!$, we only need to find the $\attackerWinMin$ relation as a least fixed point of the inductive description.
This is done by \refAlgo{alg:game-algorithm}.
Every time a new way of winning a position for the attacker is discovered, this position is added to the $\varname{todo}$.
Initially, these are the positions where the defender is stuck.
The update at an attacker position in \refLine{code:attacker-pos-update} takes the inversely updated budgets ($\energyUpdateInv$) of successor positions to be tentative attacker winning budgets.
At a defender position, the attacker only wins if they have winning budgets for all follow-up positions (\refLine{code:option-covering}).
Any supremum of such budgets covering all follow-ups will be winning for the attacker (\refLine{code:defender-pos-update}).
At both updates, we only select the minima as a finite representation of the infinitely many attacker budgets.

\subsection{Complexity and How to Flatten It}
\label{subsec:complexity}

For finite games, \refAlgo{alg:game-algorithm} is sure to terminate in exponential time of game graph branching degree and dimensionality.

\begin{lemma}[Winning budget complexity, proof see
  \ifthenelse{\boolean{fullversion}}{page~\pageref{prf:complexity}}{\cite{bisping2023equivalenceEnergyGamesReport}}]
  \label{lem:complexity}
  For an $N$-dim\-en\-sio\-nal declining energy game with $\gameMove$ of branching degree $o$, \refAlgo{alg:game-algorithm} terminates in $\bigo(\relSize{\gameMove} \cdot \relSize{G}^N \cdot (o + \relSize{G}^{(N - 1) \cdot o}))$ time, using $\bigo(\relSize{G}^{N})$ space for the output.
\end{lemma}

\begin{lemma}[Full spectroscopy complexity]
  \label{lem:spectroscopy-complexity}
  Time complexity of computing winning budgets for the full spectroscopy energy game $\gameSpectroscopy$ is in $2^{\bigo(\relSize{\proc} \cdot 2^{\relSize{\proc}})}$.
\end{lemma}
\begin{proof}
  Out-degrees $o$ in $\gameSpectroscopy$ can be bounded in $\bigo(2^{\relSize{\proc}})$,
  the whole game graph $\relSize{\gameMove_\triangle} \in \bigo( \relSize{\step{\cdot}} \cdot 2^{\relSize{\proc}} + \relSize{\proc}^2 \cdot 3^{\relSize{\proc}})$,
  and game positions $\relSize{G_\triangle} \in \bigo(\relSize{\proc} \cdot 3^{\relSize{\proc}})$.
  Insert with $N=6$ in \refLem{lem:complexity}.
  Full proof in~%
  \ifthenelse{\boolean{fullversion}}
    {the appendix on page~\pageref{prf:spectroscopy-complexity}}
    {\cite{bisping2023equivalenceEnergyGamesReport}}.
\end{proof}

\noindent
We thus have established the approach to be double-exponential.
The complexity of the previous spectroscopy algorithm~\cite{bjn2022decidingAllBehavioralEqs} has not been calculated.
One must presume it to be equal or higher as the game graph has Bell-numbered branching degree and as the algorithm computes formulas, which entails more options than the direct computation of energies.
This is what lies behind the introduction's observation that moderate nondeterminism already renders \cite{bjn2022decidingAllBehavioralEqs} unusable.

Our present energy game reformulation allows us to use two ingredients to do way better than double-exponentially when focussing on the common linear-time--branching-time spectrum:

First, \refSubsec{subsec:more-clever} has established that most of the partitionings in attacker conjunction moves can be disregarded by looking at the initial actions of processes.

Second, Fahrenberg et al.~\cite{fjls2011energyGamesMulti} have shown that considering just ``capped'' energies in a grid $\energies_K = \{0,\ldots,K\}^N$ can reduce complexity.
Such a \emph{flattening of the lattice} turns the space of possible energies into constant factor $(K+1)^N$ (with $(K+1)^{N-1}$-sized antichains) independent of input size.
For \refAlgo{alg:game-algorithm}, space complexity needed for $\varname{attacker\_win}$ drops to $\bigo(\relSize{G})$ and
time complexity to $\relSize{\gameMove} \cdot 2^{\bigo(o)}$.
If we are only interested in finitely many notions of equivalence as in the case of \refFig{fig:ltbt-spectrum}, we can always bound the energies to range to the maximal appearing number plus one.
The last number represents all numbers outside the bound up to infinity.

\begin{lemma}[Clever spectroscopy complexity]
  \label{lem:spectroscopy-complexity-clever}
  Time complexity of computing winning budgets for the clever spectroscopy energy game $\gameSpectroscopyClever$ with capped energies is in $2^{\bigo(\relSize{\proc})}$.
\end{lemma}
\begin{proof}
  Out-degrees $o$ in $\gameSpectroscopyClever$ can be bounded in $\bigo(\relSize{\proc})$,
  the whole game graph $\relSize{\gameMove_\blacktriangle} \in \bigo( \relSize{\step{\cdot}} \cdot 2^{\relSize{\proc}} + \relSize{\proc}^2 \cdot 2^{\relSize{\proc}})$,
  and game positions $\relSize{G_\blacktriangle} \in \bigo(\relSize{\proc} \cdot 2^{\relSize{\proc}})$. Inserting in the flattened version of \refLem{lem:complexity} yields:
  \begin{align*}
    \bigo(\relSize{\gameMove_\blacktriangle} \cdot 2^{C_0 \cdot o})\;
  =\; & \bigo((\relSize{\step{\cdot}} \cdot 2^{\relSize{\proc}} + \relSize{\proc}^2 \cdot 2^{\relSize{\proc}}) \cdot 2^{C_1 \cdot \relSize{\proc}})\\
  =\; & \bigo((\relSize{\step{\cdot}} + \relSize{\proc}^2) \cdot 2^{C_2 \cdot \relSize{\proc}})\\
  =\; & \bigo(\relSize{\step{\cdot}} \cdot 2^{C_2 \cdot \relSize{\proc}}).
  \end{align*}
\end{proof}
Deciding trace equivalence in nondeterministic systems is PSPACE-hard and will thus take at least exponential time.
Therefore, the exponential time of the ``clever'' spectroscopy algorithm restricted to a finite spectrum is about as good as it may get, asymptotically speaking.

\subsection{Equivalences and Distinguishing Formulas from Budgets}
\label{subsec:obtaining-eqs}

For completeness, let us briefly flesh out how to actually obtain equivalence information from the minimal attacker winning budgets $\attackerWinMin(\attackerPos{p, \{q\}})$ we compute.

\begin{definition}
  For an antichain $\mathit{Mn} \subseteq \energies$ characterizing an upper part of the energy space, the complement antichain
  $\overline{\mathit{Mn}} \defEq \Min \;
    ( \energies_\infty \cap 
    (\{ (\sup E') - (1,\ldots,1) \mid E' \subseteq \mathit{Mn} \} \cup
    \{ e(i) \in \energies_\infty \mid (e(i))_i = (\inf \mathit{Mn})_i - 1 \land \forall j \neq i \ldotp (e(i))_j = \infty \}) )$
  has the complement energy space as its downset.
\end{definition}

\noindent
$\defenderWinMax(\attackerPos{p, \{q\}}) = \overline{\attackerWinMin(\attackerPos{p, \{q\}})}$ characterizes \emph{all} preordering formula languages and thus equivalences defined in terms of expressiveness prices for $p$ and $q$.
This might contain multiple, incomparable, notions from the spectrum.
Taking both directions, $\overline{\attackerWinMin(\attackerPos{p, \{q\}}) \cup \attackerWinMin(\attackerPos{q, \{p\}})}$, will thus characterize the finest intersection of equivalences to equate $p$ and $q$.

If we just wonder which of the equivalences from the spectrum hold, we may establish this more directly by checking which of them are not dominated by attacker wins.

From the information, we can also easily build witness relations to certify that we return sound equivalence results.
In particular, the pairs won with arbitrary attacker budgets, $\{(p,q) \mid (\infty,\infty,\infty,\infty,\infty,\infty) \in \defenderWinMax(\attackerPos{p, \{q\}})\}$ are a bisimulation.
Similarly, the strategy formulas of \refDef{def:strategies} can directly be computed to explain inequivalence.

If we use symbolic winning budgets capped as proposed at the end of \refSubsec{subsec:complexity}, the formula reconstruction will be harder and the $\overline{\attackerWinMin(\attackerPos{p, \{q\}})}$ might be below the maximal defender winning budgets if these exceed the bound.
But this will not matter as long as we choose a cap beyond the natural numbers that characterize our spectrum.

\section{Exploring Minimizations}
\label{sec:implementation}

\begin{table}[b]
  \caption{\label{tab:benchmark-results}Sample systems, sizes, and benchmark
  results.}
  \centering
  \pgfkeys{/pgf/number format/.cd,
    fixed, precision=2, 1000 sep={\text{,}}}
  \addtolength{\tabcolsep}{0.9pt}
  \begin{adjustbox}{center}
    \begin{filecontents*}{data/benchmark.csv}
system, states, transitions, bisimtime, bisimquot, initialpairs, spectrotime, spectropos, spectromoves, enabledness,trace,simulation, bisimulation,spectrotimeold,spectromovesold,spectrotimenonclever,spectromovesnonclever
peterson,19,138,24,19,47,153,664,2363,3,11,11,19,,,23307,348474
vasy_0_1,289,1224,23,9,30,17,279,566,1,9,9,9,169,1118,23,1334
vasy_1_4,1183,4464,98,28,88,24,548,1000,8,28,28,28,52,1125,19,1320
vasy_5_9,5486,9392,206,145,232,61,1856,2988,109,145,145,145,138,3789,48,4315
vasy_8_24,8879,24411,855,416,2625,2148,60095,145965,171,415,415,416,540963,513690,10478,725113
vasy_8_38,8921,38424,1562,219,837,191,7523,14958,112,218,218,218,781,19595,214,19690
vasy_10_56,10849,56156,3477,2112,239878,174589,1949428,6012676,13,2112,2112,2112,,,,
vasy_18_73,18746,73043,3587,4087,172428,,,,,,,,,,,
vasy_25_25,25217,25216,575,25217,25217,329,25217,0,25217,25217,25217,25217,1152,100866,322,0
cwi_1_2,1952,2387,210,1132,481735,384125,8510136,22723369,9,1132,1132,1132,,,,
cwi_3_14,3996,14552,1041,62,1832,300,11156,18350,2,62,62,62,2481,14761,284,25666
\end{filecontents*}
\pgfplotstabletypeset[
      col sep=comma,
      header=true,
      column type=r,
      columns={system, states, bisimquot, spectromovesold, spectrotimeold, spectromovesnonclever, spectrotimenonclever, spectromoves, spectrotime, enabledness, trace, simulation},
      every head row/.style={ before row=\toprule,after row=\midrule},
      every even row/.style={ before row={\rowcolor[gray]{0.9}}},
      every last row/.style={ after row=\bottomrule},
      columns/system/.style={column name={\bf system}, string type, column type={l|}, postproc cell content/.style={@cell content={\texttt{##1}}}},
      columns/states/.style={column name={$\proc$}},
      columns/bisimquot/.style={column name={\pgfutilensuremath{\proc_{/\bEquiv{B}}}}, column type/.add={}{|}},
      columns/spectromovesold/.style={column name={\cite{bjn2022decidingAllBehavioralEqs}-\pgfutilensuremath{\gameMove}}},
      columns/spectrotimeold/.style={column name={\bf t/s}, divide by=1000, column type/.add={}{|}},
      columns/spectromovesnonclever/.style={column name={\pgfutilensuremath{\gameMove_\triangle}}},
      columns/spectrotimenonclever/.style={column name={\bf t/s}, divide by=1000, column type/.add={}{|}},
      columns/spectromoves/.style={column name={\pgfutilensuremath{\gameMove_\blacktriangle}}},
      columns/spectrotime/.style={column name={\bf t/s}, divide by=1000, column type/.add={}{|}},
      columns/enabledness/.style={column name={\pgfutilensuremath{\proc_{/\bEquiv{E}}}}},
      columns/trace/.style={column name={\pgfutilensuremath{\proc_{/\bEquiv{T}}}}},
      columns/simulation/.style={column name={\pgfutilensuremath{\proc_{/\bEquiv{1S}}}}},
      ]{data/benchmark.csv}
  \end{adjustbox}
\end{table}

Our algorithm can be used to analyze the equivalence structure of moderately-sized real-world transition systems. In this section, we take a brief look at its performance on the VLTS (``very large transition systems'') benchmark suite~\cite{garavel2017vlts} and return to our initial Peterson example.

The energy spectroscopy algorithm has been added to the Linear-Time--Branching-Time Spectroscope of~\cite{bjn2022decidingAllBehavioralEqs} and can be tried on transition systems at \url{https://equiv.io/}.

Table~\ref{tab:benchmark-results} reports the results of running the implementation of~\cite{bjn2022decidingAllBehavioralEqs} and this paper's implementation in variants using the spectroscopy energy game $\gameSpectroscopy$ and the clever spectroscopy energy game $\gameSpectroscopyClever$.
We tested on the VLTS examples of up to 25,000~states and the Peterson example (\refFig{fig:ts-peterson}).
The table lists the $\proc$-sizes of the input transition systems and of their bisimilarity quotient system $\proc_{/\bEquiv{B}}$.
The spectroscopies have been performed on the bisimilarity quotient systems by constructing the game graph underneath positions comparing all pairs of enabledness-equivalent states.
The middle three groups of columns list the resource usage for the three implementations using: the \cite{bjn2022decidingAllBehavioralEqs}-spectroscopy, the energy game $\gameSpectroscopy$, and the clever game $\gameSpectroscopyClever$.
For each group, the first column reports traversed game size, and the second gives the time the spectroscopy took in seconds.
Where the tests ran out of memory or took longer than five minutes (in the Java Virtual Machine with 8~GB heap space, at 4~GHz, single-threaded), the cells are left blank.
The last three columns list the output sizes of state spaces reduced with respect to enabledness $\bEquiv{E}$, traces $\bEquiv{T}$, and simulation $\bEquiv{1S}$---as one would hope, all three algorithms returned the same results.

From the output, we learn that the VLTS examples, in a way, lack diversity:
Bisimilarity $\bEquiv{B}$ and trace equivalence $\bEquiv{T}$ mostly coincide on the systems (third and penultimate column).

Concerning the algorithm itself, the experiments reveal that the computation time grows mostly linearly with the size of the game move graph.
Our algorithm can deal with bigger examples than~\cite{bjn2022decidingAllBehavioralEqs} (which fails at \texttt{peterson}, \texttt{vasy\_10\_56} and \texttt{cwi\_1\_2}, and takes more than 500 seconds for \texttt{vasy\_8\_24}).
Even where \cite{bjn2022decidingAllBehavioralEqs} has a smaller game graph (e.g.\ \texttt{cwi\_3\_14}), the exponential formula construction renders it slower.
Also, the clever game graph $\gameMove_\blacktriangle$ indeed is much smaller than $\gameMove_\triangle$ for examples with a lot of nondeterminism such as \texttt{peterson}.

Of those terminating, the heavily nondeterministic \texttt{cwi\_1\_2} is the most expensive example.
As many coarse notions must record the nondeterministic options, this blowup is to be expected.
If we compare to the best similarity algorithm by Ranzato and Tapparo~\cite{rt2010efficientSimulation}, they report their algorithm SA to tackle \texttt{cwi\_1\_2} single-handedly.
Like our implementation, the prototype of SA~\cite{rt2010efficientSimulation} ran out of memory while determining similarity for \texttt{vasy\_18\_73}.
This is in spite of SA theoretically having optimal complexity and similarity being less complex (cubic) than trace equivalence, which we need to cover.
The benchmarks in~\cite{rt2010efficientSimulation} failed at \texttt{vasy\_10\_56}, and \texttt{vasy\_25\_25}, which might be due to 2010's tighter memory requirements (they used 2 GB of RAM) or the degree to which bisimilarity and enabledness in the models is exploited.

\section{Conclusion and Related Work}
\label{sec:conclusion}

This paper has connected two strands of research in the field of system analysis:
The strand of \emph{equivalence games on transition systems} starting with Stirling's bisimulation game~\cite{stirling1993modal,shr1995hornsatGames,cd2008gameCharacetrizations,bjn2022decidingAllBehavioralEqs} and the strand of \emph{energy games for systems of bounded resources}~\cite{em1979positionalMeanPayoff,bjk2010evassReachability,chaloupka2010zReachability2dVASS,fjls2011energyGamesMulti,amss2013parityGamesVASS,reichert2015reachabilityGamesCounters,vcdh2015complexityMeanPayoffEnergy,fh2022twoBoundedCounterGames,kh2022energyGamesResourceBounded}.

The connection rests on the insight that levels of equivalence correspond to resources available to an attacker who tries to tell two systems apart.
This parallel is present in recent work within the security domain \cite{hm2021ePassportBisim} just as much as in the first thoughts on observable nondeterminism by Hennessy and Milner~\cite{hm1980hml}.

The paper has not examined the precise relationship of the games of \refSec{sec:priced-game} to the whole zoo of
VASS, energy, mean-payoff, monotonic~\cite{ab2008monotonicGames}, and counter games.
The spectroscopy energy game deviates slightly from common multi-energy games due to $\mathtt{min}_D$-updates and due to the attacker being energy-bound (instead of the defender).
As the energies cannot be exhausted by defender moves, the game can also be interpreted as a VASS game~\cite{bjk2010evassReachability,amss2013parityGamesVASS} where the attacker is stuck if they run out of energy.
Our algorithm complexity matches that of general lower-bounded $N$-dimensional energy games~\cite{fjls2011energyGamesMulti}.
Links between our declining energy games and other games from the literature might enable slight improvements of the algorithm.
For instance, reachability in VASS games can turn polynomial~\cite{chaloupka2010zReachability2dVASS}.

In the strand of generalized game characterizations for equivalences~\cite{shr1995hornsatGames,cd2008gameCharacetrizations,bjn2022decidingAllBehavioralEqs}, this paper extends applicability for real-world systems.
The implementation performs on par with the most efficient similarity algorithm~\cite{rt2010efficientSimulation}.
Given that among the hundreds of equivalence algorithms and tools most primarily address bisimilarity~\cite{gl2022eqChecking}, a tool for coarser equivalences is a worthwhile addition.
Although our previous algorithm~\cite{bjn2022decidingAllBehavioralEqs} is able to solve the spectroscopy problem, its reliance on super-exponential partitions of the state space makes it ill-fit for transition systems with significant nondeterminism.
In comparison, our new algorithm also needs one less layer of complexity because it determines equivalences without constructing distinguishing formulas.

These advances enable a spectroscopy of systems saturated by weak transitions.
We can thus analyze weak equivalences such as in the example of Peterson's mutex.
For special weak equivalences without a strong counterpart such as branching bisimilarity~\cite{glabbeek1993ltbt}, deeper changes to the modal logic are necessary~\cite{bj2023ltbtsSilentSteps}.

The increased applicability has allowed us to exhaustively consider equivalences on the smaller systems of the widely-used VLTS suite~\cite{garavel2017vlts}.
The experiments reveal that the spectrum between trace equivalence and bisimilarity mostly collapses for the examined systems.
It may often be reasonable to specify systems in such a way that the spectrum collapses.
In a benchmark suite, however, a lack of semantic diversity can be problematic:
For instance, otherwise sensible techniques like polynomial-time reductions~\cite{cm2019ndaReductionLanguageInclusion} will not speed up language inclusion testing,
and nuances of the weak equivalence spectrum~\cite{bn2019coupledsimTacas} will falsely seem insignificant.
One may also overlook errors and performance degradations that appear only for transition systems where equal traces do not imply equivalent branching behavior.
We hope this blind spot does not affect the validity of any of the numerous studies relying on VLTS benchmarks.

\subsubsection*{Acknowledgments.}
This work benefited from discussion with Sebastian Wolf, with David N.\ Jansen, with members of the LFCS Edinburgh, and with the MTV research group at TU Berlin, as well as from reviewer comments.

\subsubsection*{Data availability.}

\ifthenelse{\boolean{fullversion}}
{The peer-reviewed shortened version of this paper (published at CAV 2023~\cite{bisping2023equivalenceEnergyGames}) can be found on
\url{https://doi.org/10.1007/978-3-031-37706-8_5}.}
{Proofs and updates are to be found in the report version of this paper~\cite{bisping2023equivalenceEnergyGamesReport}.
}
The Scala source is on GitHub: \url{https://github.com/benkeks/equivalence-fiddle/}.
A webtool implementing the algorithm runs on \url{https://equiv.io/}.
An artifact including the benchmarks is archived on Zenodo~\cite{bisping2023ltbtsZenodo3}.

\bibliographystyle{splncs04}

\bibliography{similarities}

\ifthenelse{\boolean{fullversion}}
{
\appendix

\section{Correctness and Complexity Proofs}


\begin{postponedProof}[\refLem{lem:bisimulation-game}]
  \label{prf:bisimulation-game}
  $p_0$ and $q_0$ are bisimilar
  iff the defender wins the game $\gameSpectroscopy[\attackerPos{p_0, \{q_0\}}, e_0]$ for every initial energy budget $e_0$,
  \ie if $(\infty, \infty, \infty, \infty, \infty, \infty) \in \defenderWinMax(\attackerPos{p_0, \{q_0\}})$.
\end{postponedProof}
\begin{proof}
  First note that, for $p_0$ and $q_0$ to be bisimilar, they need to be related by a bisimulation relation $(p_0, q_0) \in \rel{R}$, where $\rel{R}$ is symmetric and has the simulation property that, for every $(p,q) \in \rel{R}$, $p \step{a} p'$ implies there is a $q'$ with $q \step{a} q'$ and $(p',q') \in \rel{R}$.
  \begin{itemize}
    \item Assume $p_0$ is bisimilar to $q_0$.
      Then, there is a bisimulation relation $\rel{R}$ with $(p_0, q_0) \in \rel{R}$.
      Construct \emph{some} defender strategy $s_\defenderSubscript$, where
      \[s_\defenderSubscript(\rho \cdot \defenderPos{p,Q,Q_*}) \mathrel{\defEq}
        \begin{cases}
          \attackerPos[\land]{p,q} & \text{if there is $q \in Q$ with $(p,q) \in \rel{R}$ to pick,}\\
          \attackerPos{p,Q_*} & \text{if there is $q \in Q_*$ with $(p,q) \in \rel{R}$,}\\
          \mathrm{undefined} & \text{if no such $q$ exists.}
        \end{cases}
      \]
      We prove inductively that plays with the defender applying this strategy will only reach positions with the invariant that $(p,q) \in \rel{R}$ at $\attackerPos[\land]{p,q}$, respectively, for some $q \in Q$ at $\attackerPos{p,Q}$, and for some $q \in Q \cup Q_*$ at $\defenderPos{p,Q,Q_*}$, rendering the picks well-defined.
      \begin{itemize}
        \item Initially, with $\rho = (\attackerPos{p_0,q_0})$ the invariant holds as $(p_0, q_0) \in \rel{R}$.
        \item At steps following $\rho \cdot \attackerPos{p,Q}$, the invariant is maintained due to the simulation property of $\rel R$.
        \item At steps following $\rho \cdot \defenderPos{p,Q,Q_*}$, the invariant is maintained due to the definition of $s_\defenderSubscript$.
        \item At steps following $\rho \cdot \attackerPos[\land]{p,q}$, symmetry of $\rel R$ ensures the invariant.
      \end{itemize}
      As this shows that the defender will not get stuck following $s_\defenderSubscript$ regardless of energy,
      the defender wins with every energy, \ie, $(\infty, \infty, \infty, \infty, \infty, \infty) \in \defenderWinMax(\attackerPos{p_0, \{q_0\}})$.
    \item Assume $(\infty, \infty, \infty, \infty, \infty, \infty) \in \defenderWinMax(\attackerPos{p_0, \{q_0\}})$.
      Construct \[\rel{R} \mathrel{\defEq} \{(p,q) \mid (\infty, \infty, \infty, \infty, \infty, \infty) \in \defenderWinMax(\attackerPos{p, \{q\}})\}.\]
      $\rel{R}$ must be a bisimulation relating $p_0$ and $q_0$.
      \begin{itemize}
        \item Clearly, $(p_0, q_0) \in \rel{R}$ because of the assumption.
        \item $\rel{R}$ is symmetric.
          As the energy level is unbounded, the attacker may always move $\attackerPos{p, \{q\}} \gameMove \defenderPos{p, \{q\}, \varnothing} \gameMove \attackerPos[\land]{p, q} \gameMove \attackerPos{q, \{p\}}$.
          Therefore, the defender wins from $\attackerPos{p, \{q\}}$ iff they win from $\attackerPos{q, \{p\}}$.
        \item $\rel{R}$ is a simulation.
          With unbounded energy levels, the attacker may always move $\attackerPos{p, \{q\}} \gameMove \attackerPos{p', Q'} \gameMove \defenderPos{p', Q', \varnothing}$, where $p \step{a} p'$ and $Q \step{\ccsInm{a}} Q'$.
          There, the defender only wins if there is a $q'$ such that they win after $\defenderPos{p', Q', \varnothing} \gameMove \attackerPos[\land]{p', q'} \gameMove \attackerPos{p', \{q'\}}$.
          Therefore, the defender winning $\attackerPos{p, \{q\}}$ unboundedly and $p \step{a} p'$ implies unbounded winning of $\attackerPos{p', \{q'\}}$ for some $q'$ with $q \step{a} q'$.
      \end{itemize}
  \end{itemize}
\end{proof}

\begin{remark}
  For finite-state systems, \refProp{prop:language-prices} and \refThm{thm:correctness} would already prove \refLem{lem:bisimulation-game}.
  The preceding proof of \refLem{lem:bisimulation-game} is more general as it also goes through for infinite-state systems.
\end{remark}

\begin{postponedProof}[\refLem{lem:price-soundness}, price soundness]
  \label{prf:price-soundness}
  $\varphi \in \hmlStrategies(\attackerPos{p,Q}, e)$ implies that $\expr(\varphi) \leq e$ and that $\expr(\varphi) \in \attackerWin(\attackerPos{p,Q})$.
\end{postponedProof}
\begin{proof}
  By induction on the structure of $\varphi$ with arbitrary $p,Q,e$.
  \begin{itemize}
    \item \emph{Base case $\hmlAndS \varnothing \in \hmlStrategies(\attackerPos{p,Q}, e)$.}
      This formula must stem from a move $\attackerPos{p,Q} \gameMoveX{(0,-1,0,0,0,0)} \defenderPos{p,Q' \setminus \{*\},Q_*}$
      for some $Q_* \subseteq Q$ and $Q' \setminus \{*\} = Q \setminus Q_*$
      with $\hmlAndS \varnothing \in \hmlStrategies(\allowbreak\defenderPos{p,Q' \setminus \{*\}, Q_*},e + (0,-1,0,0,0,0))$.
      The definition of $\hmlStrategies$ allows the latter only if the defender has no moves, that is, if $Q = Q' = Q_* = \varnothing$.
      Therefore, the attacker wins at the defender position with $(0,0,0,0,0,0)$ and thus $\expr(\hmlAndS \varnothing) = (0,1,0,0,0,0) \in \attackerWin(\attackerPos{p,Q})$.
      Also, $e$ must be above or equal $(0,1,0,0,0,0)$.
    \item \emph{Case $\hmlObs{b}\varphi \in \hmlStrategies(\attackerPos{p,Q}, e)$.}
      By the definition of $\hmlStrategies$ and the game, there must be $\attackerPos{p,Q} \gameMoveX{(-1,0,0,0,0,0)} \attackerPos{p',Q'}$
      updating $e' = \energyUpdate(e, (-1,0,0,0,0,0)) = (e + (-1,0,0,0,0,0)) \in \attackerWin(\attackerPos{p',Q'})$
      such that $\varphi \in \hmlStrategies(\attackerPos{p',Q'}, e')$.
      With the induction hypothesis, we know that $\expr(\varphi) \in \attackerWin(\attackerPos{p',Q'})$ and that $\expr(\varphi) \leq e'$.
      With the definition of $\expr$, we have that $\expr(\hmlObs{b}\varphi) = \expr(\varphi) + (1,0,0,0,0,0) \leq e' + (1,0,0,0,0,0) = e$,
      and by the definition of $\attackerWin$, also $\expr(\hmlObs{b}\varphi) \in \attackerWin(\attackerPos{p,Q})$.
    \item \emph{Case $\hmlAnd{q}{Q'}\psi_q \in \hmlStrategies(\attackerPos{p,Q}, e)$ with $Q' \neq \varnothing$.}
      This must stem from a move $\attackerPos{p,Q} \gameMoveX{(0,-1,0,0,0,0)} \defenderPos{p,Q' \setminus \{*\},Q_*}$
      for some $Q_* \subseteq Q$ and $Q' \setminus \{*\} = Q \setminus Q_*$
      with $\hmlAnd{q}{Q'}\psi_q \in \hmlStrategies(\defenderPos{p,Q' \setminus \{*\}, Q_*},\allowbreak e + (0,-1,0,0,0,0))$
      such that
        $* \notin Q'$
        or (if $* \in Q'$) $\defenderPos{p,Q' \setminus \{*\}, Q_*} \gameMoveX{(\updMin{1,3},0,0,0,0,0)} \attackerPos{p,Q_*}$
          and $\psi_* \in \hmlStrategies(\attackerPos{p,Q_*},\allowbreak (\min(e_1,e_3),e_2-1,e_3,e_4,e_5,e_6))$,
      and such that, for each clause $q \in Q' \setminus \{*\}$, $\defenderPos{p,Q' \setminus \{*\}, Q_*} \gameMoveX{(0,0,0,\updMin{3,4},0,0)} \attackerPos[\land]{p,q}$,
      as well as $(e_1,e_2-1,e_3,\min(e_3,e_4),e_5,e_6)\allowbreak \in \attackerWin(\attackerPos[\land]{p,q})$,
      and $\psi_q \in \hmlStrategies(\attackerPos[\land]{p,q}, (e_1,e_2-1,e_3,\min(e_3,e_4),e_5,e_6))$.
      Now let us examine all $\psi_q$ for $q \in Q'$:
      \begin{itemize}
        \item For $\psi_*$ (if $* \in Q'$), the induction hypothesis gives that $\expr(\psi_*) \in \attackerWin(\attackerPos{p,Q_*})$ and that $\expr(\psi_*) \leq e_* = (\min(e_1,e_3),e_2-1,e_3,e_4,e_5,e_6)$. Also, $\psi_*$ must be a positive observation because of the definition of $\hmlStrategies$.
        \item If $\psi_q$ is a negation, $\psi_q = \hmlNeg \varphi_q$,
          it must be due to $\attackerPos[\land]{p,q} \gameMoveX{(\updMin{1,5},0,0,0,0,-1)} \attackerPos{q,\{p\}}$,
          with $e_q$ equalling $(\min(e_1,e_5),e_2-1,e_3,\min(e_3,e_4),e_5,e_6-1) \in \attackerWin(\attackerPos{q,\{p\}})$
          and $\varphi_q \in \hmlStrategies(\attackerPos{q,\{p\}}, e_q)$ being an observation.
          By the induction hypothesis, $\expr(\varphi_q) \in \attackerWin(\attackerPos{q,\{p\}})$ and $\expr(\varphi_q) \leq e_q$.
        \item If $\psi_q$ with $q \neq *$ is positive, $\psi_q = \varphi_q$,
          it must be due to a move like $\attackerPos[\land]{p,q} \gameMoveX{(\updMin{1,4},0,0,0,0,0)} \attackerPos{p,\{q\}}$,
          with $e_q = (\min(e_1,e_3,e_4),e_2-1,e_3,\allowbreak\min(e_3,e_4),e_5,e_6) \in \attackerWin(\attackerPos{p,\{q\}})$
          and $\varphi_q \in \hmlStrategies(\attackerPos{p,\{q\}}, e_q)$ being an observation.
          By the induction hypothesis, $\expr(\varphi_q) \in \attackerWin(\attackerPos{p,\{q\}})$ and $\expr(\varphi_q) \leq e_q$.
      \end{itemize}
      We know by the definition of $\expr$ and with the obtained inequalities that
      \begin{align*}
        & \expr(\hmlAnd{q}{Q'}\psi_q) \\ = &
        \begin{pmatrix}
          \sup ( \{ \expr_1(\psi_q) \mid q \in Q' \})\\
          1 + \sup ( \{ \expr_2(\psi_q) \mid q \in Q' \})\\
          \sup ( \{ \expr_3(\psi_q) \mid q \in Q' \} \cup \{ \expr_1(\psi_q) \mid \psi_q \text{ positive}\} )\\
          \sup ( \{ \expr_4(\psi_q) \mid q \in Q' \} \cup \{ \expr_1(\psi_q) \mid q \in Q' \setminus \{*\} \land \psi_q \text{ positive}\})\\
          \sup ( \{ \expr_5(\psi_q) \mid q \in Q' \} \cup \{ \expr_1(\psi_q) \mid q \in Q' \setminus \{*\} \land \psi_q \text{ negative}\})\\
          \sup ( \{ 1 + \expr_6(\varphi_q) \mid \psi_q = \hmlNeg \varphi_q \} \cup \{ \expr_6(\psi_q) \mid \psi_q \text{ positive} \})
        \end{pmatrix}
        \\ \leq &
        \begin{pmatrix}
          \sup ( \{ (e_q)_1 \mid q \in Q' \})\\
          1 + \sup ( \{ (e_q)_2 \mid q \in Q' \})\\
          \sup ( \{ (e_q)_3 \mid q \in Q' \} \cup \{ (e_q)_1 \mid \psi_q \text{ positive}\} )\\
          \sup ( \{ (e_q)_4 \mid q \in Q' \} \cup \{ (e_q)_1 \mid q \in Q' \setminus \{*\} \land \psi_q \text{ positive}\} )\\
          \sup ( \{ (e_q)_5 \mid q \in Q \} \cup \{ (e_q)_1 \mid q \in Q' \setminus \{*\} \land \psi_q \text{ negative}\} ) \\
          \sup ( \{ 1 + (e_q)_6 \mid \psi_q = \hmlNeg \varphi_q \} \cup \{ (e_q)_6 \mid \psi_q \text{ positive} \})
        \end{pmatrix}
        \\ = &
        \begin{pmatrix}
          \sup ( \{ \min(e_1, e_3) \mid * \in Q'\} \!\cup\! \{ \min(e_1, e_3, e_4) \mid \psi_q \text{ positive}\}\! \cup \! \{ \min(e_1, e_5) \mid \psi_q \text{ negative}\} )\\
          1 + e_2 - 1\\
          \sup ( \{ e_3 \} \cup \{ \min(e_1, e_3) \mid \psi_q \text{ positive}\} )\\
          \sup ( \{ e_4 \} \cup \{ \min(e_1, e_3, e_4) \mid q \in Q' \setminus \{*\} \land \psi_q \text{ positive}\} )\\
          \sup ( \{ e_5 \} \cup \{ \min(e_1, e_5) \mid \psi_q \text{ negative}\} )\\
          \sup ( \{ 1 + e_6 - 1 \mid \psi_q \text{ negative}\} \cup \{ e_6 \} )\\
        \end{pmatrix}
        \\ \leq &
        \begin{pmatrix}
          e_1\\
          e_2\\
          e_3\\
          e_4\\
          e_5\\
          e_6
        \end{pmatrix}
      \end{align*}
      and, from the definition of $\attackerWinMin$, that the expressiveness price at the defender position is a supremum of the budgets that are winning in the next moves so that $\expr(\hmlAnd{q}{Q'}\psi_q) \in \attackerWin(\attackerPos{p,Q})$.
  \end{itemize}
\end{proof}

\begin{postponedProof}[\refLem{lem:price-completeness}, price completeness]
  \label{prf:price-completeness}
  $e_0 \in \attackerWin(\attackerPos{p_0,Q_0})$ implies there are elements in $\hmlStrategies(\attackerPos{p_0,Q_0}, e_0)$.
\end{postponedProof}
\begin{proof}
  $e_0 \in \attackerWin(\attackerPos{p_0,Q_0})$ means that the attacker has a winning strategy $s_\attackerSubscript$ such that all plays consistent with the strategy that start at energy level $e_0$ lead to a position where the defender is stuck.
  Without loss of generality, we may suppose that the attacker strategy does not use conjunction challenges immediately after conjunction answers or revivals. (If it did, we could transform it to a strategy where the attacker would make up their mind at the first conjunction of a sequence.)
  Consider the tree of such $s_\attackerSubscript$-plays, in particular the recurring attacker nodes of the form $\attackerPos{p,Q}$
  at energy level $e$, which must be $e = \energyLevel_{\rho \cdot \attackerPos{p,Q}} \in \attackerWin(\attackerPos{p,Q})$ after play $\rho$.
  Let us induct over this tree, proving $\hmlStrategies(\attackerPos{p,Q}, e)$ to be nonempty at each node.
  \begin{itemize}
    \item \emph{Case $s_\attackerSubscript(\rho \cdot \attackerPos{p,Q}) = \attackerPos{p', Q'}$.}
      This means that $\attackerPos{p,Q} \gameMoveX{(-1,0,0,0,0,0)} \attackerPos{p',Q'}$
      with $p \step{b} p'$
      updating $e' = \energyUpdate(e, (-1,0,0,0,0,0)) \in \attackerWin(\attackerPos{p',Q'})$. By the induction hypothesis, there must be a $\varphi \in \hmlStrategies(\attackerPos{p',Q'}, e')$.
      Thus, $\hmlObs{b}\varphi \in \hmlStrategies(\attackerPos{p,Q}, e)$ by the definition of $\hmlStrategies$.
    \item \emph{Case $s_\attackerSubscript(\rho \cdot \attackerPos{p,Q}) = \defenderPos{p, Q \setminus Q_*, Q_*}$.}
      This leads to move $\attackerPos{p,Q} \gameMoveX{(0,-1,,0,0,0,0)} \defenderPos{p,Q \setminus Q_*, Q_*}$,
      followed by branching of attacker-won moves $\defenderPos{p,Q \setminus Q_*, Q_*}\allowbreak \gameMoveX{(\updMin{1,3},0,0,0,0,0)} \attackerPos{p,Q_*}$ for $Q_* \neq \varnothing$
        and $\defenderPos{p,Q \setminus Q_*, Q_*} \gameMoveX{(0,0,0,\updMin{3,4},0,0)} \attackerPos[\land]{p,q} \allowbreak\gameMoveX{u_q} s_\attackerSubscript(\attackerPos[\land]{p,q})$ for each $q \in Q \setminus Q_*$.
      By our assumption that $s_\attackerSubscript$ does not nest conjunctions and the argument of the first case together with the induction hypothesis, there must be an observation $\varphi_* \in \hmlStrategies(s_\attackerSubscript(\attackerPos{p,Q_*}))$ for $Q_* \neq \varnothing$ as well as a positive or negative observation $\varphi_q \in \hmlStrategies(s_\attackerSubscript(\attackerPos[\land]{p,q}))$ for each $q \in Q \setminus Q_*$.
      By the definition of $\hmlStrategies$,
      this means that either $\hmlAnd{q}{Q}\psi_q \in \hmlStrategies(\defenderPos{p,Q \setminus Q_*, Q_*}, e)$ (for $Q_* = \varnothing$) or
      $\hmlAnd{q}{Q \setminus Q_* \cup \{ *\}}\psi_q \in \hmlStrategies(\defenderPos{p,Q \setminus Q_*, Q_*}, e)$,
      with $\psi_q = \hmlNeg\varphi_q$ if $u_q = (\updMin{1,5},0,0,0,0,-1)$ and $\psi_q = \varphi_q$ otherwise.
      These must also be included in $\hmlStrategies(\attackerPos{p,Q}, e)$.
  \end{itemize}
\end{proof}

\begin{postponedProof}[\refLem{lem:distinction-soundness}, distinction soundness]
  \label{prf:distinction-soundness}
  Every $\varphi \in \hmlStrategies(\attackerPos{p,Q}, e)$ distinguishes $p$ from every $q \in Q$.
\end{postponedProof}
\begin{proof}
  By induction on the structure of $\varphi$ with arbitrary $p,Q,e$.
  \begin{itemize}
    \item \emph{Case $\hmlObs{b}\varphi \in \hmlStrategies(\attackerPos{p,Q}, e)$.}
      By the definition of $\hmlStrategies$ and the game, there must be a move $\attackerPos{p,Q} \gameMoveX{(-1,0,0,0,0,0)} \attackerPos{p',Q'}$
      with $p \step{b} p'$
      and $Q \step{\ccsInm{b}} Q'$
      updating $e' = e + (-1,0,0,0,0,0)$
      such that $\varphi \in \hmlStrategies(\attackerPos{p',Q'}, e')$.
      By induction hypothesis, $\varphi$ distinguishes $p'$ from every $q' \in Q'$,
      \ie, $p' \in \hmlSemantics{\varphi}{}{}$ and $Q' \cap \hmlSemantics{\varphi}{}{} = \varnothing$. By the semantics of HML, then also $p \in \hmlSemantics{\hmlObs{b}\varphi}{}{}$ and $Q \cap \hmlSemantics{\hmlObs{b}\varphi}{}{} = \varnothing$.
    \item \emph{Case $\hmlAnd{q}{Q'}\psi_q \in \hmlStrategies(\attackerPos{p,Q}, e)$.}
      This is due to a move $\attackerPos{p,Q} \gameMoveX{(0,-1,0,0,0,0)} \defenderPos{p,Q' \setminus \{*\},Q_*}$
      for some $Q_* \subseteq Q$ and $Q' \setminus \{*\} = Q \setminus Q_*$
      with $\hmlAnd{q}{Q'}\psi_q \in \hmlStrategies(\defenderPos{p,Q' \setminus \{*\}, Q_*},\allowbreak e + (0,-1,0,0,0,0))$
      such that
        $* \notin Q'$
        or there is a move $\defenderPos{p,Q' \setminus \{*\}, Q_*} \gameMoveX{(\updMin{1,3},0,0,0,0,0)} \attackerPos{p,Q_*}$
          with $\psi_* \in \hmlStrategies(\attackerPos{p,Q_*},\allowbreak (\min(e_1,e_3),e_2-1,e_3,e_4,e_5,e_6))$,
      and such that, for each $q \in Q' \setminus \{*\}$, $\defenderPos{p,Q' \setminus \{*\}, Q_*} \gameMoveX{(0,0,0,\updMin{3,4},0,0)} \attackerPos[\land]{p,q}$,
      $(e_1,e_2-1,e_3,\min(e_3,e_4),e_5,e_6) \in \attackerWin(\attackerPos[\land]{p,q})$
      and $\psi_q \in \hmlStrategies(\attackerPos[\land]{p,q}, (e_1,e_2-1,e_3,\min(e_3,e_4),e_5,e_6))$.
      Now let us examine all $\psi_q$:
      \begin{itemize}
        \item If there is a $\psi_*$, it is from $\hmlStrategies(\attackerPos{p,Q_*}$.
          By the induction hypothesis, $p \in \hmlSemantics{\varphi_*}{}{}$ and $Q_* \cap {\hmlSemantics{\varphi_*}{}{}} = \varnothing$.
        \item If $\psi_q$ is a negation clause, $\psi_q = \hmlNeg \varphi_q$,
          it is due to $\attackerPos[\land]{p,q} \gameMoveX{(\updMin{1,5},0,0,0,0,-1)} \attackerPos{q,\{p\}}$,
          and $\varphi_q \in \hmlStrategies(\attackerPos{q,\{p\}}, (\min(e_1,e_5),e_2-1,\min(e_3,e_4),e_4,e_5,\allowbreak e_6-1))$.
          By the induction hypothesis, we know that $q \in \hmlSemantics{\varphi_q}{}{}$ and not $p \in \hmlSemantics{\varphi_q}{}{}$.
        \item If $\psi_q$ is a positive formula, $\psi_q = \varphi_q$,
          it must be due to $\attackerPos[\land]{p,q} \gameMoveX{(\updMin{1,4},0,0,0)} \attackerPos{p,\{q\}}$,
          and $\varphi_q \in \hmlStrategies(\attackerPos{p,\{q\}}, (\min(e_1,e_3,e_4),e_2-1,\min(e_3,e_4),e_4,\allowbreak e_5,e_6))$.
          By the induction hypothesis, we know that $p \in \hmlSemantics{\varphi_q}{}{}$ and not $q \in \hmlSemantics{\varphi_q}{}{}$.
      \end{itemize}
      All in all, every positive clause holds for $p$ and none of the formulas in the negative clauses holds for $p$.
      Therefore, $p \in \hmlSemantics{\hmlAnd{q}{Q'}\psi_q}{}{}$.
      For every $q \in Q \setminus Q_*$, $\psi_q$ is either a positive clause that does not hold for $q$, or a negative clause negating a formula that holds for $q$.
      If there are $q \in Q_*$, $\psi_*$ is false for all of them.
      Consequently, no $q \in \hmlSemantics{\hmlAnd{q}{Q'}\psi_q}{}{}$.
  \end{itemize}
\end{proof}

\begin{postponedProof}[\refLem{lem:distinction-completeness}, distinction completeness]
  \label{prf:distinction-completeness}
  If $\varphi$ distinguishes $p$ from every $q \in Q$, then $\expr(\varphi) \in \attackerWin(\attackerPos{p, Q})$.
\end{postponedProof}
\begin{proof}
  By induction on the structure of $\varphi$ with arbitrary $p,Q$.
  \begin{itemize}
    \item \emph{Case $\hmlObs{b}\varphi$.}
      So $p \in \hmlSemantics{\hmlObs{b}\varphi}{}{}$
      and $Q \cap \hmlSemantics{\hmlObs{b}\varphi}{}{} = \varnothing$.
      By the semantics of HML, there are $p',Q'$ with $p \step{b} p'$
      and $Q \step{\ccsInm{b}} Q'$
      such that $p' \in \hmlSemantics{\varphi}{}{}$ and $Q' \cap \hmlSemantics{\varphi}{}{} = \varnothing$.
      By induction hypothesis, $\expr(\varphi) \in \attackerWin(\attackerPos{p', Q'})$.
      Thus, as the game must contain $\attackerPos{p,Q} \gameMoveX{(-1,0,0,0,0,0)} \attackerPos{p',Q'}$,
      and using the definition of formula prices,
      $\energyUpdateInv(\expr(\varphi), (-1,0,0,0,0,0)) = \expr(\varphi) - (-1,0,0,0,0,0) = \expr(\hmlObs{b}\varphi) \in \attackerWin(\attackerPos{p, Q}))$.
    \item \emph{Case $\hmlAndS \Psi$.}
      So $p \in \hmlSemantics{\hmlAndS \Psi}{}{}$
      and $Q \cap \hmlSemantics{\hmlAndS \Psi}{}{} = \varnothing$.
      Consider all $\Psi' \subseteq \Psi$ that cover $Q$ such that there is a $\psi_q \in \Psi'$ precisely for every $q \in Q$ and
      $Q \cap \hmlSemantics{\hmlAndS \Psi'}{}{} = Q \cap \hmlSemantics{\hmlAnd{q}{Q}\psi_q}{}{} = \varnothing$.
      By the semantics of HML, $p \in \hmlSemantics{\varphi_q}{}{}$ for every positive $\psi_q = \varphi_q$, and $p \notin \hmlSemantics{\bar\varphi_q}{}{}$ for every negative $\psi_q = \hmlNeg \bar\varphi_q$.
      Also, for every $q$, there must either be a positive $\psi_q = \varphi_q$ such that $q \notin \hmlSemantics{\varphi_q}{}{}$ or a negative $\psi_q = \hmlNeg \bar\varphi_q$ such that $q \in \hmlSemantics{\bar\varphi_q}{}{}$.
      So, for every $q$, either $\varphi_q$ distinguishes $p$ from $q$ or $\bar\varphi_q$ distinguishes $q$ from $p$.
      If there are any positive $\psi_q$, consider any way of selecting a highest $\psi_q$, denote it by $\psi_*$ and denote $Q_* = Q \cap \hmlSemantics{\psi_*}{}{}$ as well as $Q' = Q \setminus Q_*$.
      With the induction hypothesis,
      $\expr(\varphi_*) = e_* \in \attackerWin(\attackerPos{p, Q_*})$
      and, for $q \in Q'$,
      each $\expr(\varphi_q) = e_q \in \attackerWin(\attackerPos{p, \{q\}})$ and $\expr(\bar\varphi_q) = \bar e_q \in \attackerWin(\attackerPos{q, \{p\}})$.
      So if there are $\varphi_q$ from positive clauses, then $\energyUpdateInv(e_q , (\updMin{1,4},0,0,0,0,0)) = \sup(e_q, ((e_q)_1, 0, 0, (e_q)_1, 0, 0)) \in \attackerWin(\attackerPos[\land]{p, q})$,
      and, for negative clauses, $\energyUpdateInv(\bar e_q ,(\updMin{1,5},0,0,0,0,-1))$ evaluates to the $\sup(\bar e_q - (0,0,0,0,0,-1), \allowbreak ((\bar e_q)_1, 0, 0, 0, (\bar e_q)_1, 0)) \in \attackerWin(\attackerPos[\land]{p, q})$.
      Both kinds of levels are reached through $(0,0,0,\updMin{3,4},0,0)$ updates.
      The attacker hence wins at the defender position $\defenderPos{p, Q', Q_*}$ regardless of the defender decisions if the energy level is at
      \begin{align*}
        e_\attackerSubscript \defEq \sup (
          & \; \{ e_q, ((e_q)_1, 0, (e_q)_1, (e_q)_1, 0, 0) \\
          & \qquad \mid q \in Q' \land \psi_q = \varphi_q \land e_q = \expr(\varphi_q) \} \\
          \cup
          & \; \{ \bar e_q - (0,0,0,0,0,-1), ((\bar e_q)_1, 0, (\bar e_q)_4, (\bar e_q)_4, (\bar e_q)_1, 0) \\
          & \qquad \mid q \in Q' \land \psi_q = \hmlNeg\bar\varphi_q \land \bar e_q = \expr(\bar\varphi_q) \} \\
          \cup
          & \; \{ e_*, ((e_*)_1, 0, (e_*)_1, 0, 0, 0) \\
          & \qquad \mid Q_* \neq \varnothing \} \; )
      \end{align*}
      or above.
      $\energyUpdateInv(e_\attackerSubscript, (0,-1,0,0,0,0)) = e_\attackerSubscript + (0,1,0,0,0,0) \in \attackerWin(\attackerPos{p, Q})$ by the game structure.
      $e_\attackerSubscript + (1,0,0,0,0,0)$ equals $\expr(\hmlAnd{q}{Q' \cup \{*\}}\psi_q)$, which must be below or equal $\expr(\hmlAndS \Psi)$ as the definition is in terms of suprema and $\Psi' \subseteq \Psi$.
      With the upward closure of winning budgets, $\expr(\hmlAndS \Psi) \in \attackerWin(\attackerPos{p,Q})$.
  \end{itemize}
\end{proof}

\begin{postponedProof}[\refThm{thm:correctness-clever}, correctness of cleverness]
  \label{prf:correctness-clever}
  Assume $e_4 \in \{0, 1, \infty\}$, $e_4 \leq e_3$, and that $e_5 > 1$ implies $e_3 = e_4$. Then, the attacker wins $\gameSpectroscopyClever[\attackerPos{p_0,Q_0}, e]$ precisely if they win $\gameSpectroscopy[\attackerPos{p_0,Q_0}, e]$.
\end{postponedProof}
\begin{proof}
  The implication from the clever spectroscopy game $\gameSpectroscopyClever$ to the full spectroscopy game $\gameSpectroscopy$ is trivial as the attacker moves in ${\gameMove_\blacktriangle}$ are a subset of those in ${\gameMove_\triangle}$ and the defender has the same moves in both games.
  For the other direction, we have to show that any move $\attackerPos{p,Q} \gameMoveX{(0,-1,0,0,0,0)}_\triangle\defenderPos{p, Q \setminus Q_*, Q_*}$ winning at energy level $e$ can be simulated by a winning move $\attackerPos{p,Q} \gameMoveX{(0,-1,0,0,0,0)}_\blacktriangle\defenderPos{p, Q \setminus Q', Q'}$.
  Without loss of generality, we assume that the attacker winning strategy for $\gameSpectroscopy$ does not immediately nest conjunctions, \ie, that the attacker will play an observation at the next occasion.
  Therefore, all $q \in Q$ that can only be beaten by negation moves must be in $Q \setminus Q_*$.
  If $e_3 = e_4$, then the attacker in the clever game can just use $\attackerPos{p,Q} \gameMoveX{(0,-1,0,0,0,0)}_\blacktriangle\defenderPos{p, Q, \varnothing}$.
  By the assumption, this leaves the cases where $e_5 \in \{0,1\}$.
  \begin{itemize}
    \item If $e_5 = 0$, the attacker cannot be using negative decisions.
    \begin{itemize}
      \item If $e_4 = 0$, the attacker winning in $\gameSpectroscopy$ means that $Q \setminus Q_*$ must equal $\varnothing$.
        Either $Q = \varnothing$, and the attacker wins immediately in $\gameSpectroscopy$ as well as in $\gameSpectroscopyClever$;
        or the attacker is moving through the circle $\attackerPos{p,Q} \gameMove_\triangle\defenderPos{p, \varnothing, Q} \gameMove_\triangle \attackerPos{p,Q}$, which does not contribute to the attacker winning and can be simulated in $\gameSpectroscopyClever$ by not moving (for now).
      \item Otherwise, $e_4 = 1$.
        We simulate by $\attackerPos{p,Q} \gameMoveX{(0,-1,0,0,0,0)}_\blacktriangle\defenderPos{p, Q \setminus Q', Q'}$ with $Q' = \{ q \in Q \mid \initials(p) \subseteq \initials(q) \}$.
        As $e$ suffices for a win in $\gameSpectroscopy$, the attacker wins after all $\defenderPos{p, Q \setminus Q', Q'} \gameMove_\blacktriangle \attackerPos[\land]{p, q}$ with $q \in Q \setminus Q' = \{ q \in Q \mid \initials(p) \not\subseteq \initials(q) \}$ moves by observing one of the $a \in \initials(p) \setminus \initials(q)$.
        For $\defenderPos{p, Q \setminus Q', Q'} \gameMove_\blacktriangle \attackerPos{p, Q'}$, let us note that necessarily $Q' \subseteq Q_*$, by contraposition:
          If $q \notin Q_*$, there must be attacker winning moves after $\defenderPos{p, Q \setminus Q_*, Q_*} \gameMove_\triangle \attackerPos[\land]{p, q}$ winning at energy level $e'$ with $e'_1 = 1$ and $e'_5 = 0$;
          these imply there is $a \in \initials(p) \setminus \initials(q)$; therefore $q \notin \{ q \in Q \mid \initials(p) \subseteq \initials(q) \} = Q'$.
        From the attacker winning $\attackerPos{p,Q} \gameMove_\triangle\defenderPos{p, Q \setminus Q_*, Q_*} \gameMove_\triangle \attackerPos{p,Q_*}$ with $e$ and $Q' \subseteq Q_*$, we may conclude that they also win $\attackerPos{p,Q'}$ in $\gameSpectroscopy$ with appropriately updated $e$,
        since shrinking the $Q$-side can only make it easier for the attacker to win.
        In case of $\attackerPos{p,Q} \gameMove_\blacktriangle\defenderPos{p, Q \setminus Q', Q'} \gameMove_\blacktriangle \attackerPos{p,Q'}$ with $e$, the attacker can thus proceed by the strategy transferred from $\attackerPos{p,Q'}$ in $\gameSpectroscopy$.
    \end{itemize}
    \item If $e_5 = 1$, the attacker may play negative one-observation subgames in conjunctions.
      By a similar argument as in the $e_4=1$ part of the previous case, we transfer wins as follows.
      \begin{itemize}
        \item For $e_4=0$, the attacker chooses $\attackerPos{p,Q} \gameMoveX{(0,-1,0,0,0,0)}_\blacktriangle\defenderPos{p, Q \setminus Q', Q'}$ with $Q' = \{ q \in Q \mid \initials(q) \subseteq \initials(p) \}$.
          For each $q \in Q \setminus Q'$, the reason is that $a \in \initials(q) \setminus \initials(p)$ can be used to win through negative observations after $\defenderPos{p, Q \setminus Q', Q'} \gameMove_\blacktriangle \attackerPos[\land]{p, q}$.
          For $\defenderPos{p, Q \setminus Q', Q'} \gameMove_\blacktriangle \attackerPos{p, Q'}$, we again establish $Q' \subseteq Q_*$ by contraposition:
            If $q \notin Q_*$, there must be attacker winning moves after $\defenderPos{p, Q \setminus Q_*, Q_*} \gameMove_\triangle \attackerPos[\land]{p, q}$ winning at energy level $e'$ with $e'_1 = 1$ and $e'_5 = 1$ through negation;
            these imply there is $a \in \initials(q) \setminus \initials(p)$; therefore $q \notin \{ q \in Q \mid \initials(q) \subseteq \initials(p) \} = Q'$.
        \item For $e_4=1$, we simply combine the previous two arguments to establish $\attackerPos{p,Q} \gameMoveX{(0,-1,0,0,0,0)}_\blacktriangle\defenderPos{p, Q \setminus Q', Q'}$ with $Q' = \{ q \in Q \mid \initials(p) = \initials(q) \}$ as appropriate.
          In particular, for each $q \in Q \setminus Q'$, there either is $a \in \initials(p) \setminus \initials(q)$ that can be observed to win after $\defenderPos{p, Q \setminus Q', Q'} \gameMove_\blacktriangle \attackerPos[\land]{p, q}$ or $a \in \initials(q) \setminus \initials(p)$ to be negatively observed.
          $\defenderPos{p, Q \setminus Q', Q'} \gameMove_\blacktriangle \attackerPos{p, Q'}$ is covered by the fact that winning attacker moves after $\defenderPos{p, Q \setminus Q_*, Q_*} \gameMove_\triangle \attackerPos[\land]{p, q}$ for $q \in Q \setminus Q_*$ through one-step observations or negated observations again justify $Q' \subseteq Q_*$.
      \end{itemize}
  \end{itemize}
\end{proof}

\begin{postponedProof}[\refLem{lem:complexity}, winning budget complexity]
  \label{prf:complexity}
  For an $N$-dimensional declining energy game with $\gameMove$ of branching degree $o$, \refAlgo{alg:game-algorithm} terminates in $\bigo(\relSize{\gameMove} \cdot \relSize{G}^N \cdot (o + \relSize{G}^{(N - 1) \cdot o}))$ time, using $\bigo(\relSize{G}^{N})$ space for the output.
\end{postponedProof}
\begin{proof}
  A position is only updated if a successor has been updated in a way leading to the discovery of winning budgets below those that already were known.
  Each energy can trigger such an update only once at each position.
  So, the first tentative set of winning budgets assigned bounds the further updates in the order of energies below it.
  These are are polynomially bounded by the size of the $N$-dimensional hypercube / grid containing these first tentative budgets and the zero-vector.
  Each first assignment is bounded in each dimension by the length of simple paths originating from a position, which can be over-approximated by $\relSize{G}$ for each position.

  Two things follow from this:
  First, the amount of antichains in the grid of occurring tentative budgets is bounded by $\relSize{G}^{N-1}$,
  leading to a space complexity due to $\varname{attacker\_win}$ of $\relSize{G} \cdot \relSize{G}^{N-1}$.
  Second, the points in the grid bound the proper updates a position can experience to $\relSize{G}^N$.
  Collectively, these can trigger at most move-many updates, so $\relSize{\gameMove} \cdot \relSize{G}^N$ bounds the updates.
  Every update must consider up to out-degree $o$ many successors.
  At defender nodes, the update may take $\relSize{G}^{(N-1) \cdot o}$ combinations into account.
  This culminates in a time complexity of $\bigo(\relSize{\gameMove} \cdot \relSize{G}^N \cdot (o + \relSize{G}^{(N - 1) \cdot o}))$.
\end{proof}

\begin{postponedProof}[\refLem{lem:spectroscopy-complexity}, full spectroscopy complexity]
  \label{prf:spectroscopy-complexity}
  Time complexity of computing winning budgets for the full spectroscopy energy game $\gameSpectroscopy$ is in $2^{\bigo(\relSize{\proc} \cdot 2^{\relSize{\proc}})}$.
\end{postponedProof}
\begin{proof}
  Out-degrees $o$ in $\gameSpectroscopy$ can be bounded in $\bigo(2^{\relSize{\proc}})$,
  the whole game graph $\relSize{\gameMove_\triangle} \in \bigo( \relSize{\step{\cdot}} \cdot 2^{\relSize{\proc}} + \relSize{\proc}^2 \cdot 3^{\relSize{\proc}})$,
  and game positions $\relSize{G_\triangle} \in \bigo(\relSize{\proc} \cdot 3^{\relSize{\proc}})$.
  Inserting with $N=6$ in \refLem{lem:complexity} yields:
  \begin{align*}
      & \bigo(\relSize{\gameMove_\triangle} \cdot \relSize{G}^N \cdot (o + \relSize{G}^{(N - 1) \cdot o}))\\
  =\; & \bigo((\relSize{\step{\cdot}} \cdot 2^{\relSize{\proc}} + \relSize{\proc}^2 \cdot 3^{\relSize{\proc}})
        \cdot (\relSize{\proc} \cdot 3^{\relSize{\proc}})^6
        \cdot (2^{\relSize{\proc}} + (\relSize{\proc} \cdot 3^{\relSize{\proc}})^{5 \cdot C_1 \cdot 2^{\relSize{\proc}}})) \\
  =\; & \bigo((\relSize{\step{\cdot}} \cdot 2^{\relSize{\proc}} \cdot \relSize{\proc}^6 \cdot 3^{6 \cdot \relSize{\proc}}
        + \relSize{\proc}^8 \cdot 3^{7 \cdot \relSize{\proc}})
        \cdot (2^{\relSize{\proc}} + \relSize{\proc}^{C_2 \cdot 2^{\relSize{\proc}}} \cdot 3^{\relSize{\proc} \cdot C_2 \cdot 2^{\relSize{\proc}}})) \\
  =\; & \bigo(
        \relSize{\step{\cdot}} \cdot 2^{2 \cdot \relSize{\proc}} \cdot \relSize{\proc}^6 \cdot 3^{6 \cdot \relSize{\proc}}
      + \relSize{\proc}^8 \cdot 3^{7 \cdot \relSize{\proc}} \cdot 2^{\relSize{\proc}} \\
    & + \relSize{\step{\cdot}} \cdot 2^{\relSize{\proc}} \cdot \relSize{\proc}^{6 + C_2 \cdot 2^{\relSize{\proc}}} \cdot 3^{6 \cdot \relSize{\proc} + \relSize{\proc} \cdot C_2 \cdot 2^{\relSize{\proc}}}
      + \relSize{\proc}^{8 + C_2 \cdot 2^{\relSize{\proc}}} \cdot 3^{7 \cdot \relSize{\proc} + \relSize{\proc} \cdot C_2 \cdot 2^{\relSize{\proc}}})\\
  =\; & \bigo(
        \relSize{\step{\cdot}} \cdot 2^{C_3 \cdot \relSize{\proc}} \cdot \relSize{\proc}^6
      + \relSize{\proc}^8 \cdot 2^{C_4 \cdot \relSize{\proc}} \\
    & + \relSize{\step{\cdot}} \cdot \relSize{\proc}^{C_5 \cdot 2^{\relSize{\proc}}} \cdot 2^{C_6 \relSize{\proc} \cdot 2^{\relSize{\proc}}}
      + \relSize{\proc}^{C_7 \cdot 2^{\relSize{\proc}}} \cdot 2^{C_8 \cdot \relSize{\proc} \cdot 2^{\relSize{\proc}}})\\
  =\; & \bigo(
      \relSize{\step{\cdot}} \cdot 2^{(C_5 \log \relSize{\proc} + C_6 \relSize{\proc}) \cdot 2^{\relSize{\proc}}})
  \end{align*}
\end{proof}
}{}

\end{document}